\documentclass[acmsmall, dvipsnames, nonacm, authorversion]{acmart}

\pdfoutput=1

\usepackage[utf8]{inputenc}
\usepackage{microtype}
\usepackage{mathtools}
\usepackage{stmaryrd}
\usepackage{MnSymbol}
\usepackage{graphicx}
\usepackage{tasks}
\usepackage{imakeidx}
\usepackage{verbatim}
\usepackage{enumitem}
\usepackage{natbib}
\usepackage{todonotes}
\ifodd\textwidth
\fi

\usepackage{thmtools, thm-restate}

\newcommand{\po}{\textcolor{BlueViolet}{\text{po}}}
\newcommand{\rf}{\textcolor{Green}{\text{rf}}}
\newcommand{\mo}{\textcolor{Red}{\text{mo}}} 
\newcommand{\hb}{\textcolor{NavyBlue}{\text{hb}}}
\newcommand{\rb}{\textcolor{RubineRed}{\text{rb}}}
\newcommand{\rfe}{\textcolor{Green}{\text{rfe}}} 
\newcommand{\rfi}{\textcolor{Green}{\text{rfi}}}

\newcommand{\psf}[3]{\textit{psafe}(#1, #2, #3)}

\newcommand{\rel}{\textcolor{CadetBlue}{\text{rel}}}

\newcommand{\comp}[2]{\text{comp}\langle #1, #2 \rangle}

\newcommand{\pts}[1]{\llangle \! #1 \! \rrangle}

\newcommand{\wf}[1]{\text{wf}(#1)}

\newcommand{\pwc}[1]{\text{pwc}(#1)}

\newcommand*{\note}[1]{\textcolor{black}{#1}}

\newtheorem{remark}{Remark}

\begin{document}

  \title{Memory Consistency and Program Transformations}

    \author{Akshay Gopalakrishnan}
    \affiliation{
      \institution{McGill University}
      \city{Montreal}
      \country{Canada}}
    \email{akshay.akshay@mail.mcgill.ca}

    \author{Clark Verbrugge}
    \affiliation{%
      \institution{McGill University}
      \city{Montreal}
      \country{Canada}}
    \email{clump@cs.mcgill.ca}

    \author{Mark Batty}
    \affiliation{%
      \institution{University of Kent}
      \city{Canterbury}
      \country{United Kingdom}}
    \email{m.j.batty@kent.ac.uk}

    \begin{abstract}
      A memory consistency model specifies the allowed behaviors of shared memory concurrent programs.
      At the language level, these models are known to have a non-trivial impact on the safety of program optimizations, limiting the ability to rearrange/refactor code without introducing new behaviors. 
      Existing programming language memory models try to address this by permitting more (\textit{relaxed/weak}) concurrent behaviors, but are still unable to allow all the desired optimizations.
      A core problem is that \textit{weaker} consistency models may also render optimizations unsafe, a conclusion that goes against the intuition of them allowing more behaviors.
      This exposes an open problem of the \textit{compositional interaction} between memory consistency semantics and optimizations; which parts of the semantics correspond to allowing/disallowing which set of optimizations is unclear.  
      In this work, we establish a formal foundation suitable enough to understand this compositional nature, decomposing optimizations into a finite set of elementary \textit{effects} on program execution traces, over which aspects of safety can be assessed.
      We use this decomposition to identify a desirable compositional property (\textit{complete}) that would guarantee the safety of optimizations from one memory model to another.
      We showcase its practicality by proving such a property between Sequential Consistency (SC) and $SC_{RR}$, the latter allowing independent read-read reordering over $SC$. 
      Our work potentially paves way to a new design methodology of programming-language memory models, one that places emphasis on the optimizations desired to be performed.
    \end{abstract}


    \maketitle

    \section{Introduction}
    \label{sec:intro}
    The semantics of shared memory concurrency is given by what is known a memory consistency model. 
    They specify constraints on the visibility of actions performed by a thread/processor to another, which give the allowed concurrent behaviors of a program.
    To be specific, it tells us what write values a read to shared memory can observe in a concurrent execution. 
    While the use of multi-core machines and software have become ubiquitous, its impact on compilers is still being uncovered.
    Compilers have traditionally played a large role in the performance of our programs, as they have been designed to automatically identify patterns in code which can be reordered/refactored to save CPU cycles/memory access times.   
    However, these program transformations (or commonly known as program/compiler optimizations) have historically been designed for sequential programs, and thus, performing them have a non-trivial consequence in a concurrent context \cite{LICM-LLVM}. 
    Even simple transformations such as constant propagation or code motion \cite{MuchCompiler} can be rendered unsafe: performing them results in the program exhibiting new behaviors. 
    The goal then, for multi-threaded program performance, is to have programming language memory consistency models that allow as many optimizations as possible.
    However, existing models used in practice are not able to satisfy this, and how to allow a desirable set of optimizations is not well understood \cite{MoiseenkoP}.

    A core reason revolves around the intuition behind how to allow a desired optimization. 
    In a sequential program (single threaded), an optimization is considered unsafe if it fails to guarantee the functional behavior of the program, a condition commonly referred to as ensuring \textit{conservative correctness} for any transformation.  
    Focusing on the possible interactions of multiple, concurrently executing code sequences, we can deem an optimization unsafe if it results
    the program exhibiting additional concurrent behaviors not intended by the programmer.
    \textit{Weak} consistency models allow more concurrent behaviors; thus, natural intuition is that \textit{weakening} the consistency semantics should be enough to allow the desired set of optimizations, necessarily preserving the safety of ones allowed by the original model. 
    However, this intuition turns out to be incorrect: weaker models may allow more behaviors, but not strictly more optimizations. 
    This complicates the understanding of how one might compose the consistency requirements of individual program optimizations or transformations that would be necessary to allow multiple optimizations.
    Our work sheds light on this issue, by answering the following two fundamental questions regarding compositionality: 
    \begin{itemize}
        \item When is the safety of an optimization implied from one model to another?  
        \item Given a memory model and a desired optimization known to be unsafe, can we formally derive/specify a memory model allowing it while retaining the safety of existing ones? 
    \end{itemize}
    
        \subsection{Core Idea} 
        
        In order to understand the compositional interaction between optimizations and consistency semantics, we observe that an optimization can be divided as a set of effects on execution traces.
        For instance, consider the example program in Fig~\ref{Core1}. 
        In each of our examples, $x,y,z$ represent shared memory and $a,b,c$ represents local memory to each thread.
        This program is transformed via code motion, moving $b=y$ above $x=1$.
        The same reordering can be translated in multiple ways as different series of effects at the execution trace level; one of which is depicted in the figure as reordering adjacent statements $b=y$, $z=1$ followed by that of $b=y$, $x=1$.     
        \begin{figure*}[htbp]
            \centering
            \includegraphics[scale=0.6]{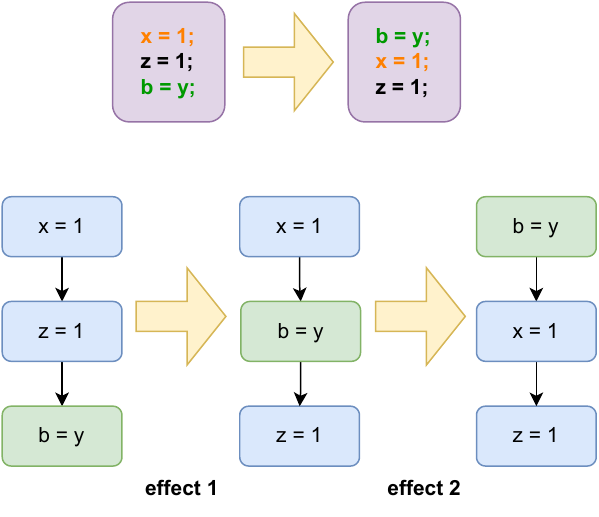}
            \caption{Code motion equivalent to performing two adjacent memory access reordering (effect 1 followed by effect 2).}
            \label{Core1}
        \end{figure*}
    
        A transformation can also result in different effects on different execution traces.
        Consider another transformation in Fig~\ref{Core2} where the program contains two possible execution traces, one for each conditional branch taken.
        The transformation of reordering $w=1$ is outside the conditional block, constitutes two separate series of effects, one for each branch taken.
        The 'then' branch can be depicted as adjacent reordering of $b=y$, $w=1$, whereas the 'else' branch can be depicted as introduction  of write $w=1$.
        \begin{figure*}[htbp]
            \centering
            \includegraphics[scale=0.6]{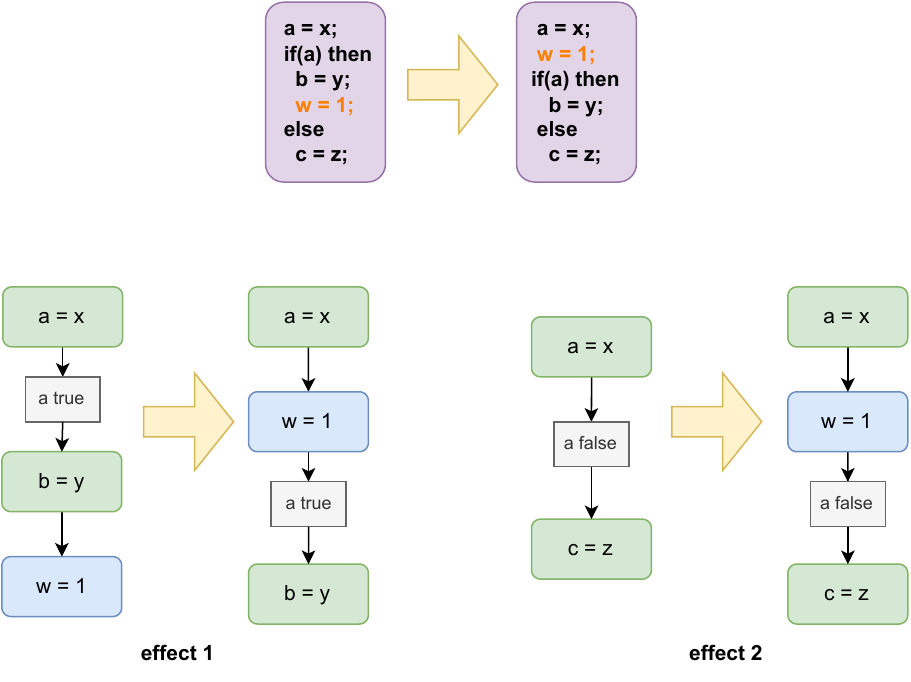}
            \caption{Code motion equivalent to reordering (effect 1) and introduction (effect 2) for two different execution traces.}
            \label{Core2}
        \end{figure*}
    
        Dividing a transformation into some combination of effects on traces of events allows us to assess the safety in terms of more elementary transformations (such as adjacent reordering), which is known \cite{MoiseenkoP} or can be proved easily \cite{SafeOptSevcik}.     
        The transformation then, is considered safe under the given consistency model, if all the corresponding effects are safe.       

    \subsection{Contributions}

        Answering the aforementioned two fundamental questions results in the following contributions:
        \begin{itemize}
            \item We represent concurrent program behaviors as a conservative approximate of execution traces, depicted by \textit{pre-traces} and \textit{candidate-executions} (Section~\ref{subsec:prelim-prog}).
            \item We use this approximation to decompose program transformations into different \textit{transformation-effects} on each pre-trace (Section~\ref{subsec:prelim-transf}). 
            \item We propose a methodology to derive consistency models that represent collections of desirable transformation-effects, followed by formally specifying the desired compositional interaction \textit{Complete}, which guarantees safety of transformation-effects from one model to another (Section~\ref{sec:methodology}).
            \item We showcase the practicality of the proposed methodology by deriving model $\textit{SC}_\textit{RR}$ using Sequential Consistency (\textit{SC}) and reordering of independent reads, followed by proving it Complete for a significant set of transformation-effects safe under SC (Section~\ref{sec:concrete}). 
        \end{itemize} 

        \paragraph*{Road Map}

        \note{
        Section~\ref{sec:motivation} goes over some informal background and motivation which can be skipped by readers familiar with the problem of safe optimizations and memory consistency.
        }
        Section~\ref{sec:prelim} goes over our view for programs as an over-approximate set of \textit{pre-traces}, followed by the formal description of program transformations as a set of \textit{transformation-effects}.
        Section~\ref{sec:methodology} showcases our proposed methodology and formalizes the desired compositional interaction \textit{Complete}.
        Section~\ref{sec:concrete} goes over concrete models, showing the practicality in deriving models by adding desired transformations and proving the model Complete w.r.t. the original.
        Section~\ref{sec:discussion} discusses the advantage of having such an approach for programming language memory models; one that focuses on the desirable set of compiler optimizations.
        Finally, Section~\ref{sec:related} goes over related work followed by Section~\ref{sec:conclusion} with conclusion and potential future work.
   
    \section{Background/Motivation}
    \label{sec:motivation}
    Sequential interleaving or Sequential Consistency (SC), is the de-facto memory model relied on while designing concurrent programs \cite{Lamport}.
    Informally, it is defined as a set of all outcomes possible by an arbitrary interleaving of the sequential executions of each thread.   
    For instance, consider the store-buffering (SB) litmus test in Fig~\ref{SB}. 
    Each shared memory is initialized to $0$.
    As per SC, several sequential interleavings are possible, each potentially giving us a different outcome. 
    The outcome $a=1 \wedge b=0$ in the green box is possible due to the interleaving order shown by the blue boxes below it.
    T1 can perform both its actions, which can be followed by actions of T2.
    The outcome in the red box, however, cannot be observed under \textit{SC}; the blue boxes below it along with the red arrows indicate a cyclic order of execution. 
    Via \textit{SC} reasoning, since $y=1$ sequentially precedes the copy of $x$ into $a$, $a$ holding the value $0$ implies $y = 1$ has happened before $x = 1$, thus enforcing $b = 1$.
    \begin{figure*}[htbp]
        \centering
        \includegraphics[scale=0.6]{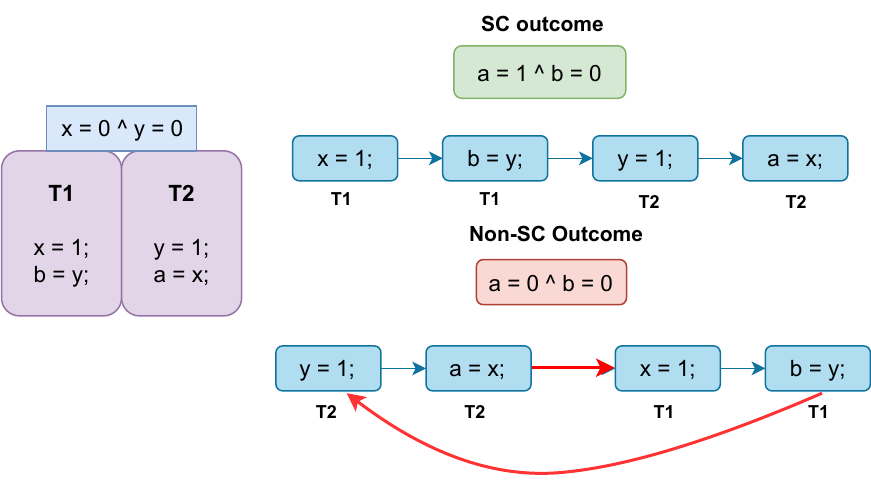}
        \caption{The store buffering (SB) litmus test showcasing an allowed (green box) and forbidden (red box) outcome under \textit{SC}.}
        \label{SB}
    \end{figure*}

    While \textit{SC} indeed is closer to a programmer's intuitive reasoning, its semantics have a detrimental impact on performance \cite{AdveS}.
    We direct our focus to the impact on compiler optimizations; syntactic changes (that preserve semantics) ranging from code-motion, common sub-expression elimination to whole algorithm transformations \cite{Dragon}, \cite{MuchCompiler}.
    The performance of the program, then, can be attributed to applying these transformations several times in different phases, which result in efficient computations overall.
    Fig~\ref{Opt} showcases some of the common optimizations done by the compiler on our programs. 
    Copy/constant propagation identifies redundant reads ($b=x$) from memory whose value already exists in the program ($x=1$).  
    Dead store elimination removes writes which have been overwritten before its value is used.
    Redundant store elimination removes writes which write the same current value to memory. 
    Loop invariant code motion moves code fragments whose result do not change in any iteration outside the loop.
    \begin{figure*}[htbp]
        \centering
        \includegraphics[scale=0.6]{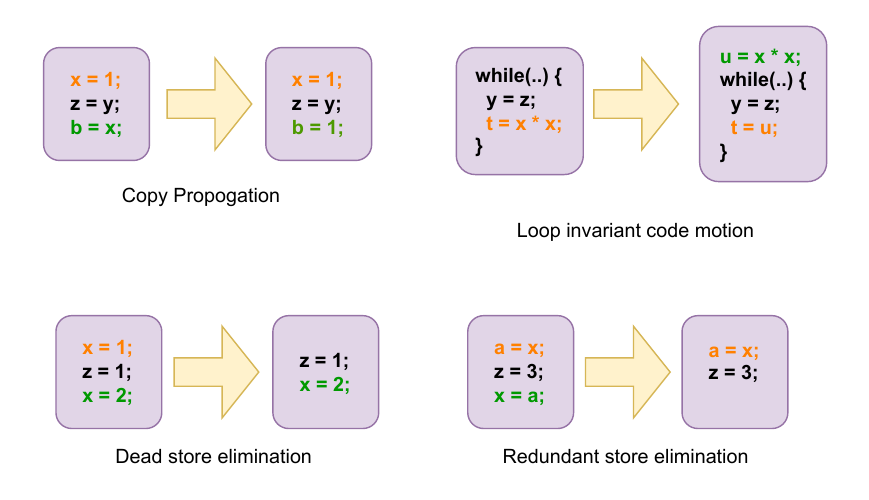}
        \caption{Different optimizations routinely performed by compilers today.}
        \label{Opt}
    \end{figure*}

    Many of these optimizations rely on plain code motion, which help in identifying redundancies in a thread-local (peephole) fashion.
    Redundancies such as these are typically identified through data-flow analysis \cite{MuchCompiler}, a technique that propagates information through all possible control flow paths of a program, with a code modification enabled due to some invariant property of the information flow between two program points.  
    For instance, the constant propagation example in Fig~\ref{Opt} is possible since x holding the value 1 is true at all points between the $x=1$ statement and the $b=x$ assignment.  
    We may also view it as a sequence of transformations based on more local relations: it is safe to reorder $b=x$ and $z=y$, and then it is easy to assert (via data-flow analysis) no side-effects exists.
    
    Traditionally (and due to the way conservative estimates of concurrent control flow rapidly waters down information), the invariants core to optimizations are established by following only sequential control flow.  
    The optimizations performed as a consequence, however, can have a significant impact on concurrent behavior.  
    For example, applying constant propagation on the store buffering example in Fig~\ref{SB-CP} may determine that $y$ holding $0$ is an invariant from the initialization of y through the entire execution of T1, and thus replacing $b=y$ with $b=0$ is valid.  
    This makes the outcome in the "red" box of Fig~\ref{SB} observable, even if SC is otherwise assumed, as shown by the sequential interleaving in blue boxes of Fig~\ref{SB-CP}.
    \begin{figure*}[htbp]
        \centering
        \includegraphics[scale=0.6]{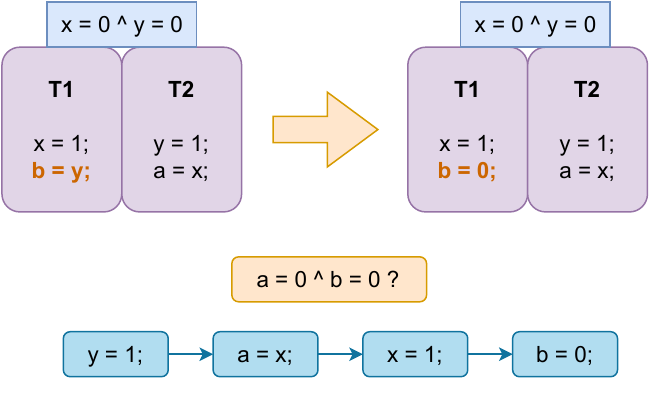}
        \caption{The forbidden outcome $a=0 \ \wedge \ b=0$ permitted under \textit{SC} after constant propagation, setting $b=0$.}
        \label{SB-CP}
    \end{figure*}
    In fact, even simple independent code motion can be rendered unsafe under \textit{SC}.
    \note{Revisiting the store buffering example in Fig~\ref{SB-WR}, the compiler can see no harm (using thread-local information) in performing independent write-read reordering.}    
    This syntactic change, however, allows the outcome originally forbidden under \textit{SC}; the blue boxes in Fig~\ref{SB-WR} show a sequential interleaving justifying the outcome.     
    \begin{figure*}[htbp]
        \centering
        \includegraphics[scale=0.6]{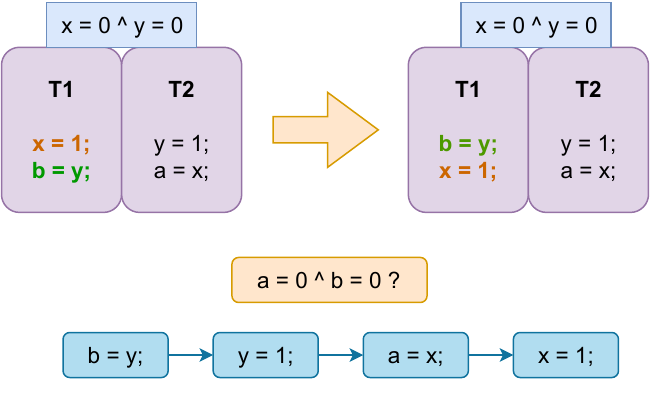}
        \caption{The forbidden outcome $a=0 \ \wedge \ b=0$ permitted under \textit{SC} after reordering $x=1$ and $b=y$.}
        \label{SB-WR}
    \end{figure*}

\subsection{Weak Consistency Models and Program Transformations}

    The impact of concurrency semantics on hardware optimizations have been studied extensively \cite{AdveS}, resulting in weaker memory models: those that allow more hardware optimizations. 
    One such example is Total Store Order (TSO), whose semantics represents an abstract machine using FIFO store buffers per-CPU.
    Informally, every write is first committed to the respective processor-specific write buffer.  
    Reads to memory are first checked from the processor's own write-buffer and a fence/lock instruction forcing the buffer to be flushed to main memory.
    Under TSO, the outcome disallowed under SC can be explained via store buffering \cite{OwensS}.
    Fig~\ref{TSO-SB} shows how the store buffering non-SC outcome is possible. 
    From the lens of TSO, both the writes $x=1$ by CPU0 and $y=1$ by CPU1 can be buffered, while the symmetric reads of $y$ by CPU0 and $x$ by CPU1 can still read from main memory, thus allowing both the reads to return $0$. 
    \begin{figure*}[htbp]
        \centering
        \includegraphics[scale=0.6]{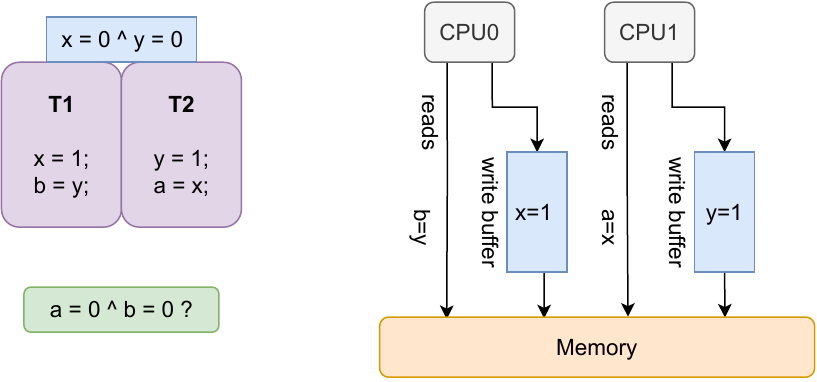}
        \caption{The forbidden \textit{SC} outcome $a=0 \ \wedge \ b=0$ for the program (left) allowed under TSO due to store buffering.}
        \label{TSO-SB}
    \end{figure*}
    
    Intuitively, the effect of TSO write buffers is similar to program transformations (optimizations) that involve independent writes and reads to be reordered. 
    Such reordering is unsafe under \textit{SC}, which may compel us to conclude \textit{TSO} allows more code transformations than \textit{SC}.
    While such reordering is indeed allowed \cite{MoiseenkoP}, the same cannot be said for rest of the optimizations.
    Let us consider for instance, thread-inlining transformation, which sequentially merges code of two threads into one.
    Although not a common or traditional optimization, and requiring some concurrency awareness, a compiler may decide to do this in order to match the amount of processors available in the target multi-core system.
    This optimization would avoid unnecessary context switching overhead \cite{ContSwtch}, thus improving performance.
    Fig~\ref{tso-invalid} represents one such instance of doing it. 
    To see why this optimization is unsafe, first note the outcome in the orange box is not observed under \textit{TSO}, as writes of $x$ and $y$ committed to main memory should be visible immediately to all processors as per its semantics (\textit{multi-copy-atomic}).
    Thus, if $T1$ observes $y=1$ before $x=1$ then $T4$ is constrained by that same ordering.
    However, if we do inline $T1$ and $T3$, as shown on the right, the outcome in question can be observed under \textit{TSO}: the read to $y$ can be taken from the write buffer of $T3;T1$ rather than main memory, and thus $T4$ is free to observe $x=1$ prior to $y=1$.
    \begin{figure*}[htbp]
        \centering
        \includegraphics[scale=0.6]{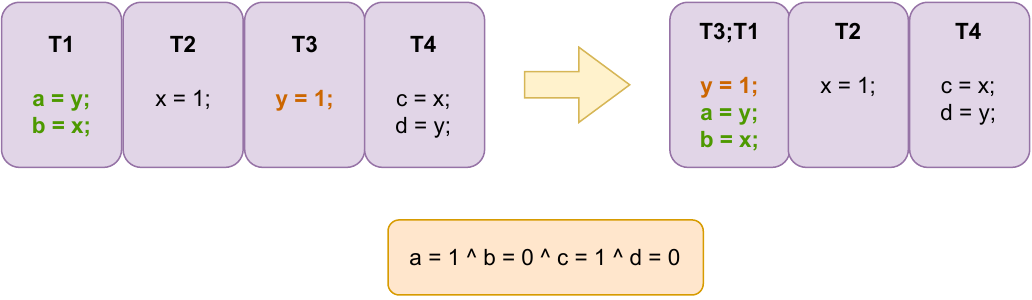}
        \caption{Example program (left) where inlining $T1$ sequentially after $T3$ (right) permits the forbidden behavior (yellow box) under \textit{TSO}, but not under \textit{SC}.}
        \label{tso-invalid}
    \end{figure*}
    In contrast, it is well known that \textit{thread inlining} is a sound transformation under \textit{SC}; it simply places a constraint on possible sequential interleavings of the program.
    
    To conclude, we now have an example of a weak memory model under which we can observe more concurrent behaviors, but does \textit{not strictly allow more safe transformations}.
    Figure~\ref{SC-TSO-Venn} shows the picture we have for \textit{TSO} and \textit{SC}.
    The Venn diagram on the left indicates the set of concurrent behaviors observed under these two models, whereas the right indicates our conclusion based on the set of transformations safe under them.
    \begin{figure*}[htbp]
        \centering
        \includegraphics[scale=0.6]{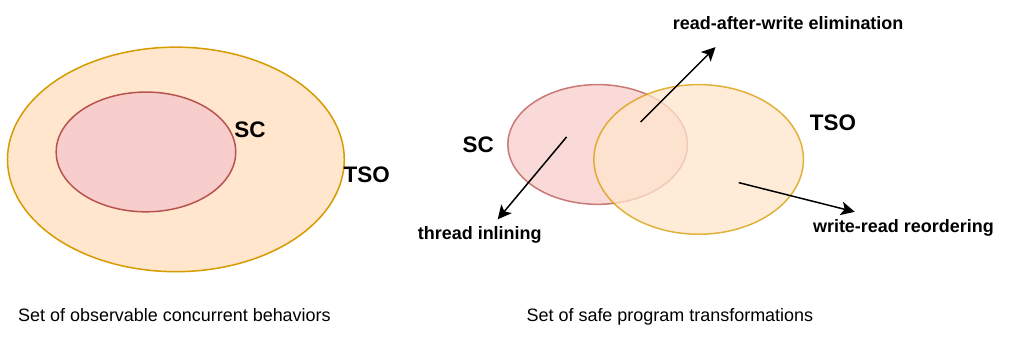}
        \caption{Venn diagram showcasing the relation between sets of allowed concurrent outcomes (left) and program transformations (right).}
        \label{SC-TSO-Venn}
    \end{figure*}

    \section{Preliminaries}
    \label{sec:prelim}
    We first go over our view of programs, representing program behaviors as a conservative approximate of \textit{pre-traces} and their \textit{candidate executions} (Sec~\ref{subsec:prelim-prog}).
    This is intended to represent a compiler's conservative view of program executions.   
    We then formally specify memory consistency models, adopting their declarative (axiomatic) format to filter out these execution traces (Sec~\ref{subsec:mem-model}).
    This is followed by defining a \textit{transformation-effect} over pre-traces, into which we decompose any program transformation (Sec~\ref{subsec:prelim-transf}).
    We then formalize safety of any transformation-effects, followed by linking it to safety of program transformations (Sec~\ref{subsec:prelim-safety}).  
    We finally discuss briefly the role of sequential (thread-local) semantics in proving safety of a given program transformation (Sec~\ref{subsec:seq-sem}).

    \subsection{Programs, Traces and Executions}
    
        \label{subsec:prelim-prog}

        We view concurrent programs ($\textit{prog}$) as a parallel composition ($||$) of sequential programs ($\textit{sp}$), each of which is associated with a thread id $t$ and a sequence of actions $p$.
        Each $p$ is composed of statements ($st$), conditional branches (if-then-else).
        Each $st$ is either a shared memory event or a local computation ($a=e$).  
        Memory events can either be a read from (eg: $a=x$) or a write to (eg: $x=v$) shared memory.
        Write events are also associated with a value $v$, which can, in general be a constant value or a local variable; for simplicity we limit $v$ to integers in this exposition.  
        The complete program grammar is as shown below. 
        \begin{align*}
            &prog := sp || prog \ | \ sp \\
            &sp := t:p \\ 
            &p := st \ | \ p;p \ | \ \text{if}(cond) \ \text{then} \ \{p\} \ \text{else} \ \{p\} \\ 
            &st := a\!=\!x; \ | \ x\!=\!v; \ | \ a\!=\!e \\ 
            &e := a \ | \ v \\
            &cond := \text{true} \ | \ \text{false} \ | \ a == v \ | \ a !\!= v \\   
            &\text{domains} := v \in \mathbb{Z} \ | \ t \in (\mathbb{N}\!\cup\!\{0\})
        \end{align*}

        \begin{remark}
            Following previous work on this topic \cite{Dodds}, we assume all loops terminate in a finite number of iterations and require static thread construction (via the $||$ operator).
            Our proofs/results however, do not depend on this assumption. 
        \end{remark}
        
        We represent initial values of shared memory as a sequence of write events syntactically ordered before all other memory accesses. 
        Let $\textit{tid}$ and $\textit{mem}$ be mapping of an event to the thread id and memory (addresses, although in examples we will use unique variable names).
        Given a program in this way, we construct an abstract program graph.
        \begin{definition}
            An \emph{abstract program} $Pr$ is an abstraction of the original program, retaining the syntactic order represented by the binary relation $\po_{prog}$, set of shared memory events $St_{prog}$ and conditional branch points.
        \end{definition}
        
        Fig~\ref{Prog-to-Abs} is an example of a program and its abstraction.
        Let $st_{prog}(Pr)$ return the set of statements and $\po_{prog}(Pr)$ return the set of syntactic orders from abstract program $Pr$.
        Let $St_{t}$ and $\po_{t}$ represent the set of events and syntactic order of thread $t$.
        Note that $\po_{\text{prog}}$ is a strict partial order.
        \begin{figure*}[htbp]
            \centering
            \includegraphics[scale=0.6]{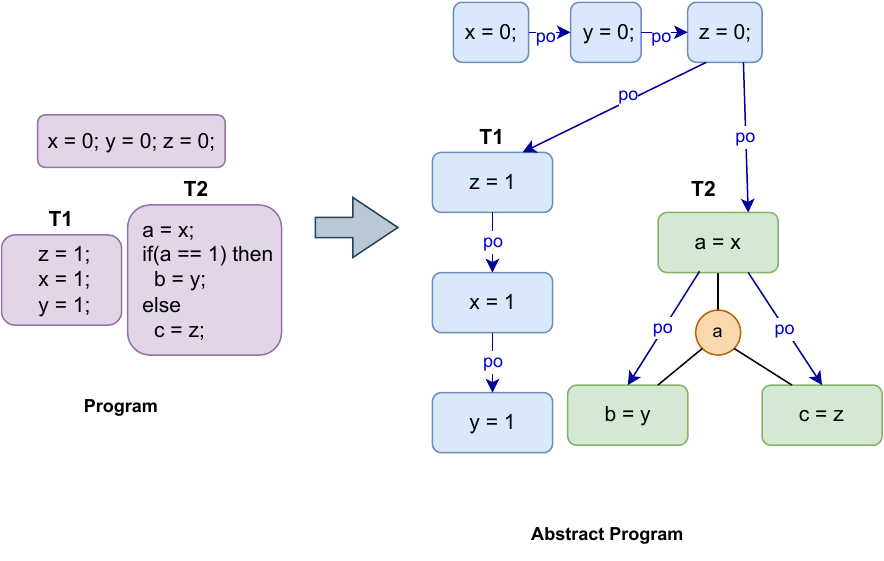}
            \caption{Example program (left) mapped to its abstract form (right).}
            \label{Prog-to-Abs}
        \end{figure*}

        We now take an abstract program $Pr$ and construct a pre-trace. 
        A pre-trace represents a possible execution trace consisting of maximal events from an Abstract Program and their syntactic order.
        We derive this by syntactically tracing program execution paths for each thread, non-deterministically choosing a path on every conditional branch point.
        \note{Formally, we first define $St_t^{\textit{trace}}  \subseteq\ St_t$ as a subset of program statements of thread $t$ such that $\forall a,b \neq a \in St_t^{trace}$ we have:
        \begin{align*}   
            &\langle a, b \rangle \in \po_{t}(Pr) \vee \langle b, a \rangle \in \po_{t}(Pr). \\                          
            &\langle a, c \rangle \in \po_{t}(Pr) \implies c \in St_{t}^{trace}.\\  
            &\langle c, b \rangle \in \po_{t}(Pr) \implies c \in St_{t}^{trace}.     
        \end{align*}
        }
       
        We then derive $po_{t}^{trace} \subseteq \po_{t}$ that represent the syntactic order between events in $St_{t}^{trace}$.
        Lastly, let $\po_{init}$ represent the syntactic order between initialization writes and memory events of each thread with $tid > 0$. 
        \begin{definition}
            \label{def:pre-trace}
            A \emph{pre-trace} ($P$) is a tuple consisting for some \emph{trace} of $Pr$, the $St_{t}^{trace}$ and appropriate $\po_{t}^{trace}$ for each thread.
            \begin{align*}
                P = ( \bigcup_{t=0}^{n}St_{t}^{trace}, \ \bigcup_{t=1}^{n}\po_{t}^{trace} \cup po_{init}).
            \end{align*}
        \end{definition}
        Fig~\ref{Pr-to-P} shows a mapping of an abstract program to all its pre-traces.

        \begin{remark}
            Since $\po_{t}^{trace}$ and $\po_{init}$ are strict total orders, we omit the subscripts for $\po$ in pre-traces. 
        \end{remark}

        Let $st(P)$ return the set of statements and $\po(P)$ return the set of program orders.
        Further, let $r(P)$ and $w(P)$ return the set of read and write memory events of $P$ respectively.
        Let $st_{loc}(P), r_{loc}(P), w_{P}$ return the set of events of $P$ operating on memory location $loc$.
        Let $!loc$ be those other than $loc$ events.
        \begin{figure*}[htbp]
            \centering
            \includegraphics[scale=0.6]{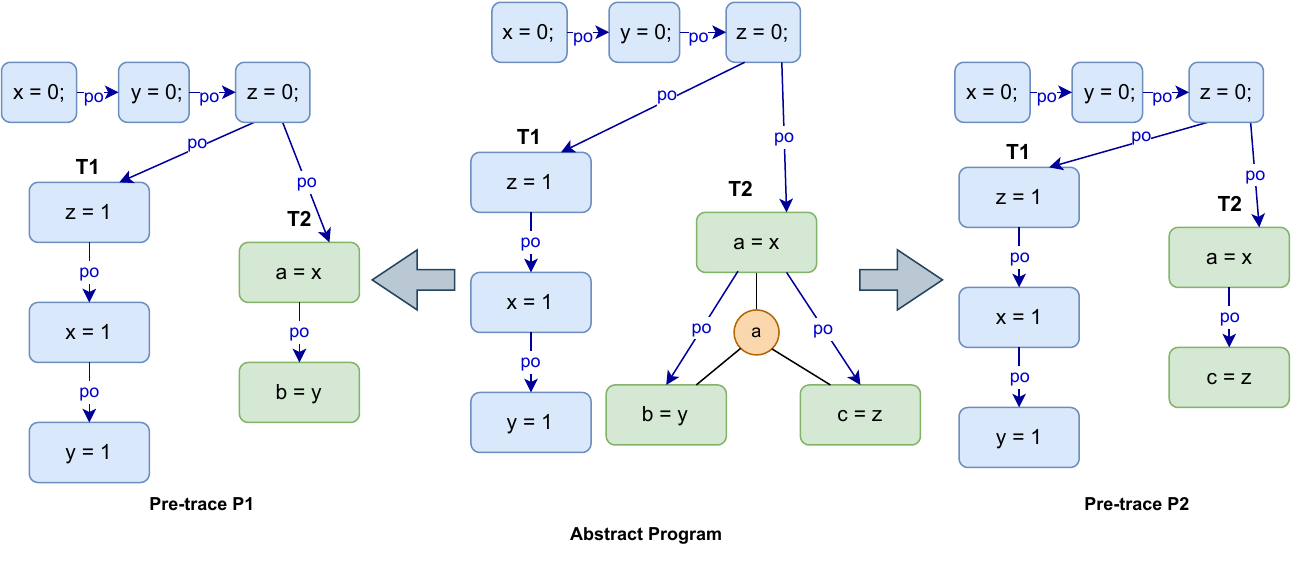}
            \caption{Example abstract program (middle) mapped to its two possible pre-traces (left and right).}
            \label{Pr-to-P}
        \end{figure*}
   
        \note{
        Pre-traces still do not represent an outcome, which is generally considered to be the final state of memory (both local as well as shared). 
        To represent outcomes, we define \textit{read-from} ($\rf$) to be a binary relation from a write event to a read event of same memory and \textit{memory-order} ($\mo$) to be the propagation order between write events of $P$.
        Let $\mo_{loc}$ represent the memory order between writes to the same memory and $max(loc)$ give the maximal write (if exists) in $\mo_{loc}$.
        }
        \note{ 
        \begin{definition}
            \label{def:exec}
            An execution is a pre-trace $P$ augmented with non-empty $\rf$, $\mo$ relations between events.
            \begin{align*}
                E = (P, \rf, \mo). 
            \end{align*}
        \end{definition}
        Let $p(E)$, $\rf(E)$, $\mo(E)$ be projection functions that return the pre-trace $P$, $\rf$ and $\mo$ relations respectively.
        We say $E$ is \emph{well-formed} ($\wf{E}$) if every read event is associated with an $\rf$ relation. 
        \begin{definition}
            \label{def:cand-exec}
            An execution $E$ is a candidate (or a \emph{candidate execution} $E$) if 
            \begin{tasks}(2)
                \task $\wf{E}$.
                \task $\rf^{-1}$ \text{functional}.
                \task $\mo$ \text{total order}.
                \task $\forall loc \ . \ max(loc) \neq \phi$.
            \end{tasks}
            where \emph{loc} represents all possible shared memory locations\footnotemark. 
        \end{definition}
        We finally define an outcome to be the $\rf(E)$ and $\mo(E)$ extracted from any candidate execution $E$.  
        }

        \footnotetext{
            The existence of a maximal event in every $\mo_{loc}$ captures the final state of memory.
        }

        Let $\langle P \rangle$ represent the set of candidate executions $E$ of given pre-trace $P$.
        Fig~\ref{P-to-E} shows one possible candidate execution for a given pre-trace $P$.
        \begin{figure*}[htbp]
            \centering
            \includegraphics[scale=0.6]{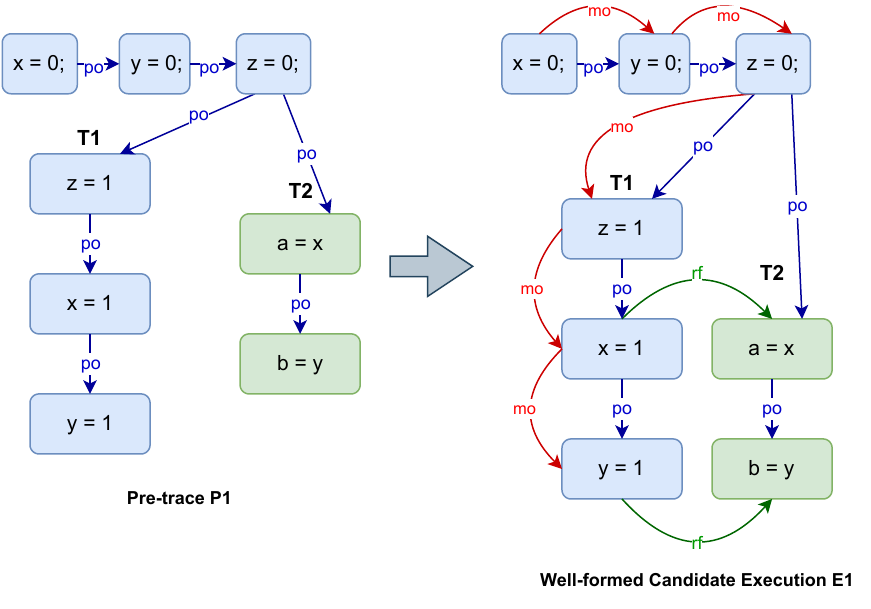}
            \caption{Example pre-trace (left) mapped to a possible candidate execution (right).}
            \label{P-to-E}
        \end{figure*}

        \begin{remark}
            Unless explicitly stated, we refer to $E$ as a candidate execution.
        \end{remark}


    \subsection{Memory Model}

    \label{subsec:mem-model}

    We have decomposed programs into a finite set of pre-traces and candidate executions. 
    We now define a memory model, which Informally, is a set of rules that apply to candidate executions. 
    \note{
    \begin{definition}
        \label{def:mem-model}
        A \emph{memory consistency model} $M$ is a finite set of consistency rules. 
        A \emph{consistency rule} $a \in M$ is a function that maps executions to boolean values. 
    \end{definition}
    We utilize the memory model to give us executions (and thereby candidate executions) that adhere to its rules.
    \begin{definition}
        \label{def:consistent}
        An execution $E$ is \emph{consistent} w.r.t. memory model $M$, written $c_{M}(E)$, if all the consistency rules of $M$ return \textit{true}\footnotemark. 
        \begin{align*}
            c_{M}(E) \ \implies \ \forall a \in M \ . \ a(E) = \textrm{true}. 
        \end{align*}
    \end{definition}
    }
  
    \footnotetext{
        Note that an execution need not be a \emph{candidate} to be consistent under a memory model.
    }
    An execution $E$ is \emph{inconsistent} w.r.t. memory model $M$ if at least one consistency rule returns \textit{false}.     
    As an example, suppose a memory model $M$ has a single consistency rule: one that returns true for any execution when  $\po \cup \mo \ \text{acyclic}$.
    Fig~\ref{ex-mem-model} shows two candidate executions, one consistent and other inconsistent as per $M$.
    \begin{figure*}[htbp]
        \centering
        \includegraphics[scale=0.6]{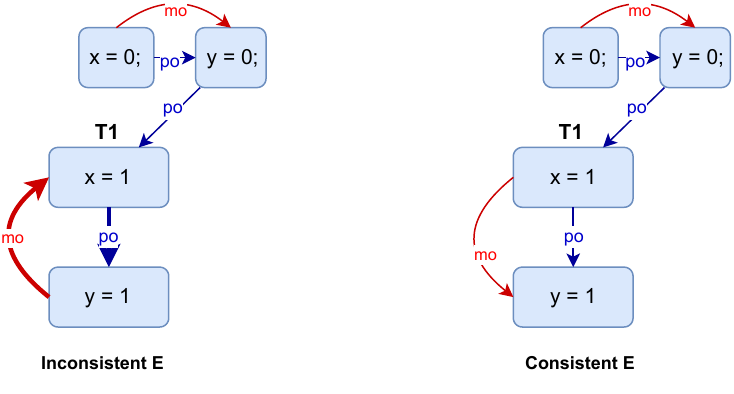}
        \caption{Inconsistent (left) and Consistent (right) executions as per memory model with consistency rule $(\po \cup \mo) \ \text{acyclic}$.}
        \label{ex-mem-model}
    \end{figure*}

    \begin{remark}
        In all our examples of candidate executions, we will omit showing certain relations which can be implied.
        For instance, we do not show all the $\mo$ and $\po$ edges in Fig~\ref{ex-mem-model}.   
    \end{remark}

    Let $\llbracket P \rrbracket_{M}$ represent the set of candidate executions of pre-trace $P$ given memory model $A$.
    Let $I\langle P \rangle_{M}$ represent the set of inconsistent candidate executions. 
    We then have $I\langle P \rangle_{M} \cup \llbracket P \rrbracket_{M} = \langle P \rangle$.    
    
    \paragraph*{Actual Behaviors of a Program}
    Note that the memory model is not enough to assert whether an outcome of the program is possible.
    A memory model needs to be coupled with thread-local language semantics to filter out pre-traces and corresponding candidate executions which are not possible by a program.
    For instance, Fig~\ref{pretrace-seq-sem} shows a pre-trace of a program tracing the false conditional path: one that cannot happen. 
    These pre-traces can be filtered out by thread-local reasoning.
    Once the relevant set of pre-traces are determined, we can filter out the program's behavior as per the entire language semantics.
    \begin{remark}
        Our results are independent of the choice of thread-local semantics adopted.
    \end{remark}
    \begin{figure*}[htbp]
        \centering
        \includegraphics[scale=0.6]{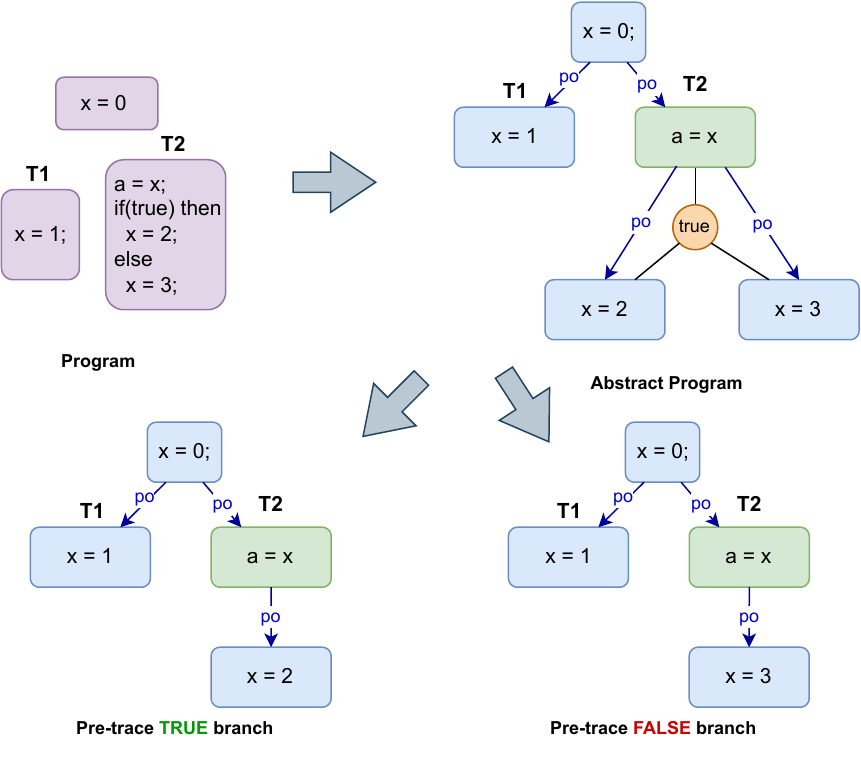}
        \caption{The pre-trace generated for program (top left) with \textit{false} branch (bottom right) will be filtered out by appropriate thread-local semantics.}
        \label{pretrace-seq-sem}
    \end{figure*}

    \subsection{Program Transformations}

    \label{subsec:prelim-transf}

    The set of compiler optimizations we aim to address can be placed under the umbrella of program transformations: syntactic changes to programs with the intention to preserve semantics.
    To relate these transformations to our view of programs, we note that each transformation at the program level has an effect on pre-traces, which in turn affect candidate executions.
    For instance, revisiting a previous example in Fig~\ref{fig:reord-eg}, reordering $w=1$ outside the conditional block at the program level affects the pre-traces in two different ways: one as simple reordering and one as write introduction. 
    More generally, these effects can result in adding/removal of statements $st$ and modification to syntactic order ($\po$).
    \begin{definition}
        \label{def:transf-effect}
        A \emph{transformation-effect} $tr$ is defined on a pre-trace $P \mapsto_{tr} P'$ as
        \begin{align*}
            P' = ((st(P)  - st^{-})\cup st^{+} , \ (\po(P) - \po^{-})\cup \po^{+} ). 
        \end{align*}
        where $st^{+}$, $\po^{+}$ represent elements added and $st^{-}$, $\po^{-}$ those removed.
    \end{definition}

    Going back to Fig~\ref{fig:reord-eg}, the reordering has the transformation-effects $P1 \mapsto_{tr} P1'$, $P2 \mapsto_{tr} P2'$ defined as 
    \begin{align*}
        &P1' = (st(P1), \ (\po(P1) - \{[b=y];\po;[w=1]\}) \cup \{[w=1];\po;[b=y]\} ). \\
        &P2' = (st(P2) \cup \{ w=1 \}, \ \po(P2) \cup (\{[a=x];\po;[w=1]\} \cup \{[w=1];\po;[b=y]\}) ). 
    \end{align*} 
    
    \begin{remark}
        Since $\po$ is transitive, we omit mentioning certain edges added/removed as part of a transformation-effect that are implied via transitivity. 
    \end{remark}

    \begin{figure*}[htbp]
        \centering
        \includegraphics[scale=0.6]{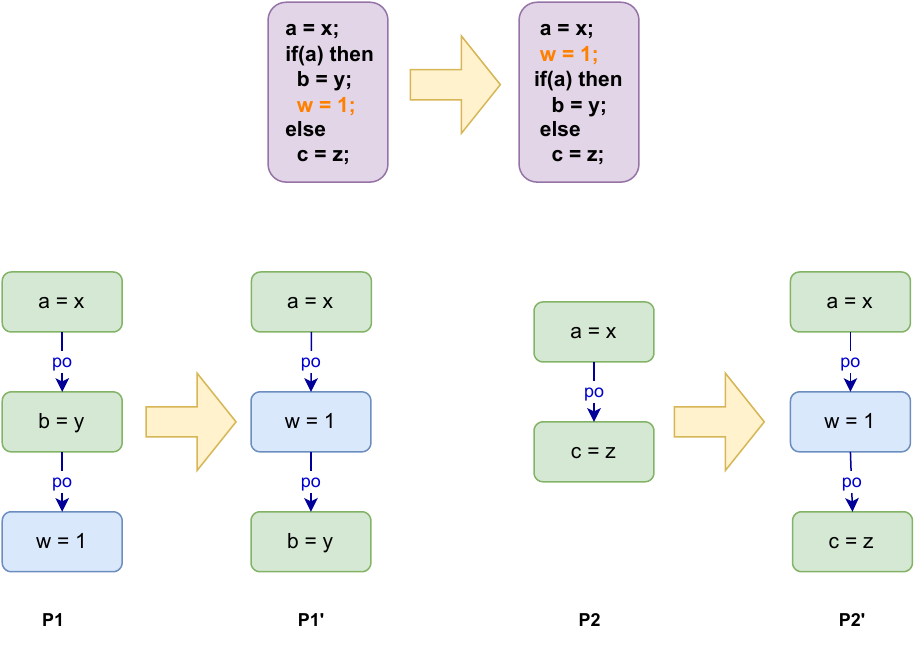}
        \caption{Example of reordering $w=1$ outside the conditional block in program as transformation-effects $P1 \mapsto P1'$ and $P2 \mapsto P2'$.}
        \label{fig:reord-eg}
    \end{figure*}

    \subsubsection{Identifying Relevant Pre-trace Pairs}

        We mentioned above that a program transformation, in essence, has an effect on pre-traces.
        Def~\ref{def:transf-effect} is generic; any pre-trace can be considered as a result of a transformation-effect on any other pre-trace.
        The example in Fig~\ref{fig:reord-eg} however, specifically has just two effects $P1 \mapsto P1'$ and $P2 \mapsto P2'$. 
        While we label these as transformed pre-traces, it is instead a mapping from pre-traces of one program to pre-traces of the transformed program. 
        This implies the appropriate effects that constitute a program-transformation are not naive mapping of every pre-trace of the original program to every pre-trace of the transformed.
        More specifically, we want to ensure the transformation is described by effects $P1 \mapsto P1'$ and $P2 \mapsto P2'$, but not $P1 \mapsto P2'$ or $P2 \mapsto P1'$.
        
        \paragraph*{Tracking event set identities}
            To identify these appropriate pairs of pre-traces, we need a way to compare them.
            We start by associating each memory event with a unique identity.
            This would ensure that a transformation-effect can be derived by simply comparing the event set of two pre-traces; making it easier to identify which statements were removed/reordered/added.
            Having such meta-data is similar to the meta-data LLVM associates to the program, which the compiler can keep track of, during an optimization \cite{MetaDataSoham}.
            
        \paragraph*{Tracking Conditional branch choices}
            
            Associating identities helps us more precisely categorize effects involving memory events that are congruent (eg: two reads to the same location x). 
            But this is not sufficient to identify the exact pairs of pre-traces in Fig~\ref{fig:reord-eg}; the meta-data on events alone would still consider $P1, P2'$ to be a valid pair.    
            
            To address this, we additionally associate the conditional branch points with a unique identity, which captures the conditional expression and an identity to the branch. 
            This ensures we know exactly which set of conditional branch points remain unchanged between the original and transformed program.
            With this additional meta-data, we ensure that for each common branch point, the same choice (then branch or else branch) is made to generate the corresponding pre-traces.
            Two branch points are common if they have the same identity and same expression. 
            This rules out the case of $P1, P2'$ being a valid pair in Fig~\ref{fig:reord-eg}.            
            \begin{remark}
                Asserting the equality of two conditional expressions can be complex, incorporating thread-local semantics to identify more common branch points, but for simplicity we assume them to be just the same expression.      
            \end{remark}
            Fig~\ref{fig:reord-track-effect} summarizes our idea of event set and conditional branch point identities. 
            The memory events are captured by identities $A, A1, A2, A3$ and the branch point is captured by $C1$.
            \begin{figure*}[htbp]
                \centering
                \includegraphics[scale=0.6]{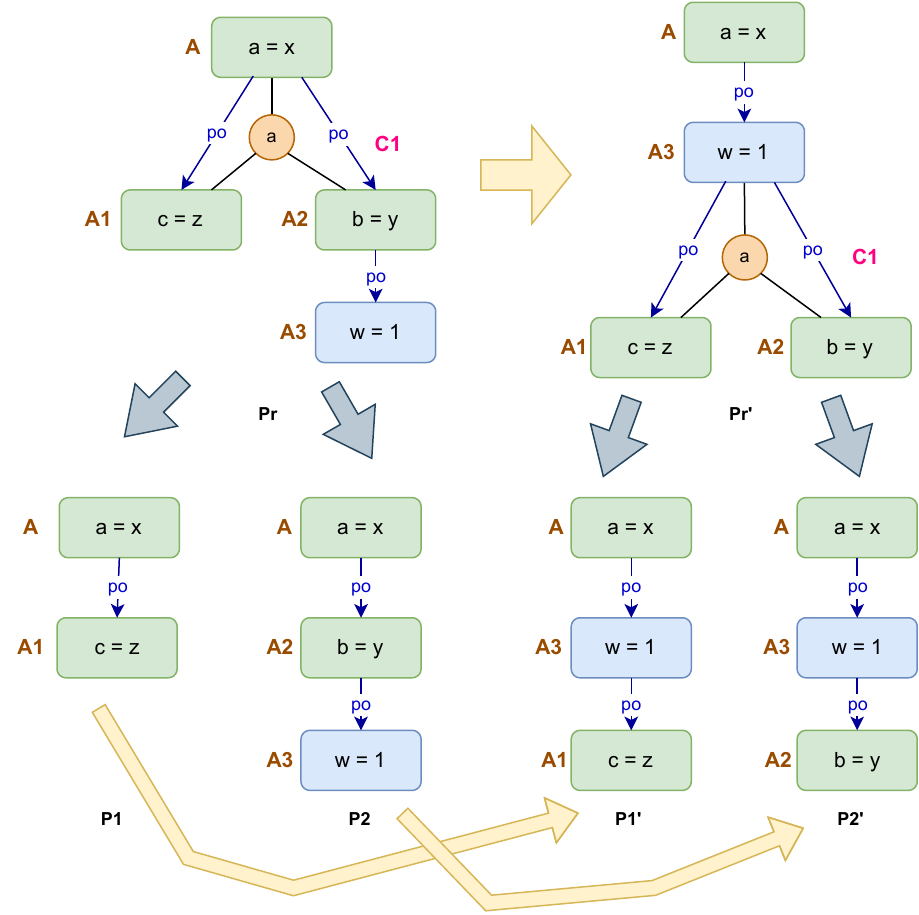}
                \caption{Tracking event set identities ($A, A1, A2, A3$) and conditional branch points ($C1$) allows us to represent the reordering as only the two effects $P1 \mapsto P1', P2 \mapsto P2'$.}
                \label{fig:reord-track-effect}
            \end{figure*}

        More generally, let $B_{i}$ represent a branch point unique identity, where $id(B_{i})$ is the label, $e(B_{i})$ the associated conditional expression and $c(B_{i})$ the choice made. 
        Let $cond(P)$ be the set of all conditional branch points $\{ B1, B2, B3 \ldots \}$ for pre-trace $P$.
        Then, two pre-traces $P1, P2$ are comparable if the choices for every conditional branch point common to both is the same.
        We now finally have a process to compare pre-traces, which form a pair between which the transformation-effect occurs.
        Let $\pts{Pr}$ represent the set of pre-traces for abstract program $Pr$. 
        \begin{definition}
            \label{def:pre-trace-cmp}
            Two pre-traces $P \in \pts{Pr}, P' \in \pts{Pr'}$ are \emph{comparable} ($P \sim P'$) if for each conditional branch point common to both, the same conditional outcome is chosen. 
            \begin{align*}
                \forall B \in cond(P), B' \in cond(P') \ . \ id(B) = id(B') \wedge e(B) = e(B') \implies c(B) = c(B').
            \end{align*}
        \end{definition}

        \begin{remark}
            Every $P \mapsto_{tr} P'$ we consider henceforth inherently implies $P \sim P'$.
        \end{remark}

        Given such similar pre-traces, we can extract the appropriate transformation-effect that changes one pre-trace to another.

    \subsection{Safety of Transformation and its Effects}
        
        \label{subsec:prelim-safety}

        On having explained how to identify the relevant pairs of pre-traces that constitute a transformation-effect, we can now delve into their safety.  
        In general, a program transformation is considered safe if the resultant program does not introduce additional behaviors.
        In our context, a behavior is synonymous to a candidate execution.
        Therefore, we first specify how to compare them.
        We can compare two candidate executions, if their common set of read events have the same $\rf$ relations and their common set of writes have the same $\mo$.
        \note{
        \begin{definition}
            \label{def:sim-exec}
            Two executions $E, E'$ are \emph{comparable}, written $E' \sim E$, if
            \begin{align*}
                &r(p(E')) \cap r(p(E)) = codom(\rf(E') \cap \rf(E)) \\ 
                &\forall a \in \mo(E) \ . \ (dom(a) \in w(p(E')) \wedge codom(a) \in w(p(E'))) \implies a \in \mo(E').  
            \end{align*}
            where \textit{dom} and \textit{codom} returns the domain and co-domain respectively of input binary relation.    
        \end{definition}
        }
        
        \begin{remark}
            $\sim$ is symmetric ($E' \sim E \iff E \sim E'$).    
        \end{remark}
        
        We now specify how to compare sets of behaviors.
        \begin{definition}
            \label{def:cand-exec-set-cmp}
            An execution set $A$ is \emph{contained} in candidate execution set $B$, written $A \sqsubseteq B$, if  
            \begin{align*}
                \forall E' \in A \ . \ \exists E \in B \ \text{s.t} \ E' \sim E.  
            \end{align*}       
        \end{definition}
        
        We first use the above relation to constrain the transformation-effects $P \mapsto_{tr} P'$ that we address; those that do not introduce a write-value.
        Excluding them is not uncommon: while introducing redundant/unused writes is safe in sequential code, they could potentially be read by other threads in a concurrent setting. 
        This can be expressed as the set of candidate executions of $P'$ to always be contained in that of $P$, or formally $\langle P' \rangle \sqsubseteq \langle P \rangle$.
        \begin{remark}
            We henceforth only consider transformation-effects that respect $\langle P' \rangle \sqsubseteq \langle P \rangle$.    
        \end{remark}
        
        We now can specify the safety of a transformation-effect.
        Informally, we require the effect to not introduce any additional consistent behaviors.
        This can be formally expressed as follows. 
        \begin{definition}
            \label{def:safe-transf}
            A transformation-effect $tr$ on a pre-trace $P \mapsto_{tr} P'$ is \emph{safe} under memory model $M$, written $\psf{M}{tr}{P}$, if 
            \begin{align*}
                &\llbracket P' \rrbracket_{M} \sqsubseteq \llbracket P \rrbracket_{M}.
            \end{align*}
        \end{definition}

        \note{
            \paragraph{Effects Altering Final Shared Memory State}
            We note that Def~\ref{def:safe-transf} would allow effects that may remove writes altering the final state of memory.
            For instance, even for a simple message-passing program (left) in Fig~\ref{fig:MP-FS}, transformation removing writes $x=1$ and $y=1$ (right) will always be considered a safe effect as per Def~\ref{def:safe-transf} under sequential consistency.
            However, in reality, this alters the final state of memory from $x=1 \wedge y=1$ (left) to $x=0 \wedge y=0$ (right), which is clearly something we do not desire.
            To prevent this, we instead assume the final writes (maximal writes in $\mo_{loc}$) can be all read at the end of the program execution\footnotemark. 
            This can be represented in the form of reads syntactically ordered after all events in a pre-trace, which are also associated with an $\rf$ relation as normal reads. 
            Having this assumption prevents us from allowing effects of the form like in Fig~\ref{fig:MP-FS}.
        }

        \begin{figure*}[htbp]
            \centering
            \includegraphics[scale=0.6]{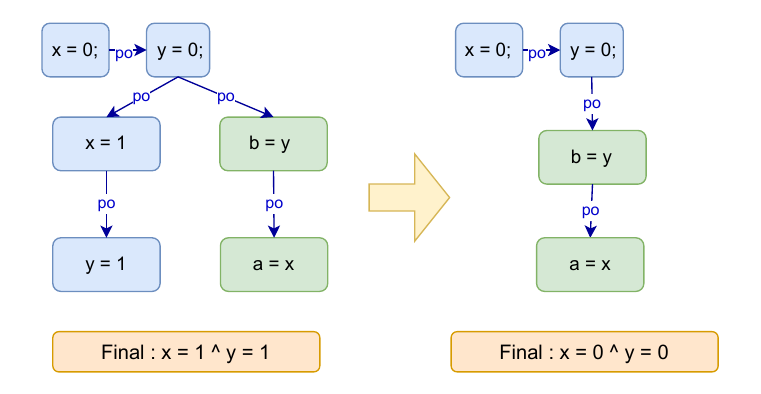}
            \caption{Def~\ref{def:transf-safe} considers the above effect to be safe under $SC$, thereby allowing a new behavior with different final state of memory.}
            \label{fig:MP-FS}
        \end{figure*}

        \footnotetext{
            The set of reads to the final state of memory can be a subset of shared memory, as some locations may not be used (read) by the program at all.
            We can leave the precise choice of final reads to the compiler. 
        }

    
    \subsubsection{From Safety of Effects to Transformations}

        Def~\ref{def:pre-trace-cmp} gives us the pairs of pre-traces for which we need to check safety of transformation-effects, while Def~\ref{def:safe-transf} gives us the criteria for safety of a transformation-effect.
        In order to determine safety of the transformation, we first note that every pre-trace of the transformed program represents one possible execution trace. 
        In order to ensure no additional behaviors are introduced, each of these pre-traces must then have at least one matching comparable pre-trace from the original program.
        Then, for all such matching pairs, we would require the corresponding effects to be safe. 
        \begin{definition}
            \label{def:transf-safe}
            A program transformation  $t$, modifying abstract program $Pr$ to $Pr'$, is safe under memory model $M$ if 
            \begin{align*}
                &\forall P' \in \pts{Pr'} \ . \ \exists P \in \pts{Pr} \ . \ P \sim P'. \\
                &\forall P \in \pts{Pr}, P' \in \pts{Pr'} \ . \ P \sim P' \implies \psf{P}{tr}{M}.
            \end{align*}
        \end{definition}
        


    \section{Methodology}
    \label{sec:methodology}
    With the base elements setup, there are two directions one can take. 
    The first obvious one is to test existing models used in practice.
    However, this would require additionally mapping the instruction semantics to a common language over which we can formally specify and compare the models for our purpose \cite{IMMAnton}.
    While this can be done, we note that our primary goal is on understanding the interaction between memory consistency semantics and optimizations.
    Hence, we resort to the second approach; we derive models incrementally by adding desired transformations known to be unsafe under a base model. 
    In this section, we lay out our methodology in deriving memory models this way, followed by formally identifying the properties desired that make it compositional w.r.t. transformations.

    We first lay out the intuition to our approach, specifying the desired compositional properties that the resultant model must satisfy w.r.t. the original (Sec~\ref{subsec:meth-proposal}).
    We then formally specify each desired property, followed by laying out in brief a strategy to prove them (Sec~\ref{subsec:meth-weaksound}, Sec~\ref{subsec:meth-comp}).
    \subsection{Proposed Approach}

        \label{subsec:meth-proposal}

        Recall that a transformation-effect is unsafe under memory model $B$ because on performing such a transformation, a new behavior may be introduced. 
        This means, the behavior in the original program was disallowed by some consistency rule (axiom/constraint) being violated.
        Hence, allowing a transformation implies removing a subset of axioms from $B$.
        Fig~\ref{fig:prop-transf-model} explains this intuition. 
        \begin{figure*}[htbp]
            \centering
            \includegraphics[scale=0.6]{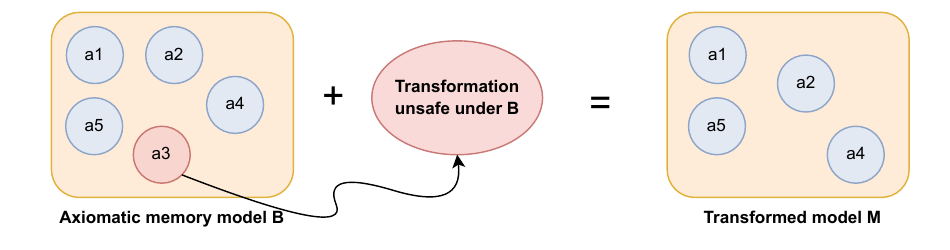}
            \caption{For memory model $B$, allowing an unsafe transformation is equivalent to removing consistency rule $a3$, giving us resultant transformed model $M$.}
            \label{fig:prop-transf-model}
        \end{figure*}
        Suppose we identify such an axiom for some transformation-effect $tr$ and remove it (ex: $a3$ in Fig~\ref{fig:prop-transf-model}), giving us resultant transformed model $M$.
        To ensure this is the desired model, we need to prove the following: 
        \begin{itemize}
            \item The resultant model $M$ allows more behaviors: \textit{Weak}.
            \item The resultant model $M$ allows the desired transformation-effect $tr$ for all pre-traces: \textit{Sound}.
            \item The resultant model $M$ preserves the safety of existing transformation-effects: \textit{Complete}.
        \end{itemize}

        \paragraph*{Relational notations}
        Given a binary relation $R$, let $R^{-1}$, $R^{?}$, $R^{+}$ represent inverse, reflexive closure and transitive closure respectively. 
        Let $[E]$ represent the identity relation on set $E$.
        Let $R1;R2$ represent left sequential composition of two binary relations.
        \note{
        Let $[A];R;[B]$ represent the relation $R \cap (A \times B)$.
        Lastly, we say the composition $R1;R2$ \emph{forms a cycle} or $R1;R2$ \emph{cycle} if there exists a cyclic path given by a non-empty relation $[a];R1;[b];R2;[a]$.
        }
        
        \paragraph*{Assumptions}

            Before we formally specify and discuss each of the above parts to prove, we first go state our assumptions on memory models. 
            Each memory model we consider will place a constraint on the existence of specific cycles in the graph. 
            \note{
                We assume that each such constraint of the memory model can be represented as a set of irreflexivity constraints over binary relations instead\footnotemark.
            }
            \footnotetext{
                Many existing models are generally represented in the form of acyclic constraints (eg: C11, RC11, Java, etc.).
                Identifying an equivalent set of irreflexivity constraints which is suitable enough for our purpose is left to future work.
            }
            For example, the constraint $(\po \cup \rf)^{+} \ \text{acyclic}$ can instead be specified as $(\po \cup \rf)^{+} \ \text{irreflexive}$. 
            Next, we assume no constraint is redundant: one cannot be contained within another. 
            For example, the constraint $(\po \cup \rf)^{+} \ \text{irreflexive}$ is contained within the constraint $(\po \cup \rf \cup \mo)^{+} \ \text{irreflexive}$\footnotemark.
            \footnotetext{
                Such a constraint being violated implies the relation forms a cycle.
            }

    \subsection{Weak and Sound}

    \label{subsec:meth-weaksound}
    A memory model is \textit{weaker} than another if it allows more concurrent behaviors.
    This can be formally expressed as follows.
    \begin{definition}
        \label{def:weak-mem-model}
        A memory model $W$ \emph{weaker} than memory model $B$ ($\text{weak}\langle W, B \rangle$) if every execution consistent in $B$ is also consistent in $W$.    
        \begin{align*}
            \forall E \ . \ c_{B}(E) \implies c_{W}(E).
        \end{align*}
    \end{definition}
    Proving weakness for our methodology is not required, as by construction, we remove an axiom from the base model. 
    
    A transformation-effect is \textit{sound} for a memory model if it is safe for any pre-trace.
    This can be formally expressed as 
    \begin{definition}
        \label{def:sound-transf}
        A transformation effect $tr$ is \emph{sound} w.r.t. memory model $M$ $sound(M, tr)$ if 
        \begin{align*}
            \forall P \ . \ \psf{A}{tr}{P}.
        \end{align*}
    \end{definition}

    Proving sound naively would require proving safety for all possible pre-traces.
    Fortunately, it suffices to quantify pre-traces using the axioms of given memory model, reasoning over their set of consistent/inconsistent candidate executions \cite{SafeOptSevcik}.

    \subsection{Complete}

    \label{subsec:meth-comp}

    In Sec~\ref{subsec:prelim-transf}, we saw how transformations can be decomposed as a series of transformation-effects.
    Therefore, in order to preserve safety of transformations from one model to another, we must preserve the safety of the composed effects.
    \begin{definition}
        \label{def:complete}
        A memory model $M$ is \emph{Complete} w.r.t. memory model $B$ ($\text{comp}\langle M, B\rangle$) if 
        \begin{align*}
            \forall P \mapsto_{tr} P' \ . \ \psf{B}{tr}{P} \implies \psf{M}{tr}{P}. 
        \end{align*}
    \end{definition}
    
    Proving Def~\ref{def:complete} requires us to quantify over all possible transformation-effects.
    One approach can be to decompose each transformation into a series of elementary effects, which can defined by us.  
    However, this can quickly turn out to be infeasible, when we desire to quantify over effects like those in Fig~\ref{fig:comp-transf}.
    Such a transformation involves several variants of de-ordering ($\po^{-} \neq \phi$), inlining ($\po^{+} \neq \phi$) and reordering (as in Fig~\ref{fig:reord-eg}) effects, depending on the way we decompose them.
    \begin{figure*}[htbp]
        \centering
        \includegraphics[scale=0.6]{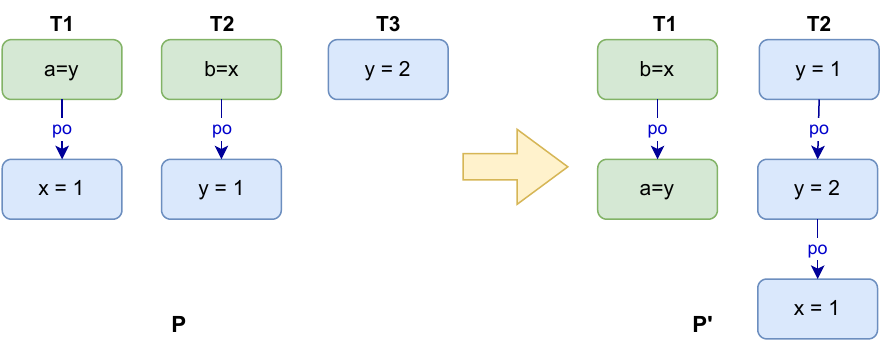}
        \caption{Example of a non-trivial transformation-effect $P \mapsto P'$, which can be decomposed into several forms of de-ordering, reordering and inlining transformation-effects.}
        \label{fig:comp-transf}
    \end{figure*}
    Separating our concerns case-wise over each set in Def~\ref{def:transf-effect} is also not practical, as even simple reordering involves both $\po^{-}$ and $\po^{+}$ in unison.
    We therefore try to prove Complete by defining the set of transformation-effects in an alternative fashion; one that utilizes the axioms of the models instead.

    We note that the set of transformation-effects we require proving Complete, by Def~\ref{def:complete}, are those safe under $B$.
    From a contrapositive lens, we instead require the set of transformation-effects $tr$ unsafe under $M$.
    Each of these $tr$ unsafe under $M$, by Def~\ref{def:safe-transf}, tell us that some consistent execution $E'$ in the transformed pre-trace has no comparable consistent execution $E$ in the original.
    The contraposition then requires the transformation-effect to also be unsafe under $B$.
    By $\text{weak}\langle M, B \rangle$, proving it would require us to handle only the case when $E'$ is also inconsistent under $B$.
    If this information suffices to prove that every such $tr$ is unsafe under $B$, then we are done.
    The above idea culminates into the following theorem.
    \begin{theorem}
        \label{thm:complete-gen}
        Consider memory models $B$, $M$ s.t. $\text{weak}\langle M, B \rangle$. 
        Consider any pre-trace $P$ and all transformation-effects $P \mapsto_{tr} P'$ such that 
        \begin{align*}
            \exists E' \in \llbracket P' \rrbracket_{M} \ . \ \forall E \in \langle P \rangle \ . \ E \sim E' \implies E \in I\langle P \rangle_{M},
        \end{align*}   
        If $\forall tr \ . \ \neg c_{B \setminus M}(E') \implies \neg \psf{B}{tr}{P}$ then $\text{comp} \langle M, B \rangle$. 
    \end{theorem}

    \begin{proof}
        We prove this by contraposition.
        Assuming $\neg comp\langle M, B \rangle$, we have
        \begin{flalign*}
            &\to \exists P \mapsto_{tr} P' . \neg (\psf{B}{tr}{P} \implies \psf{M}{tr}{P}). \\ 
            &\to \psf{B}{tr}{P} \wedge \neg \psf{M}{tr}{P}. \tag*{(Def~\ref{def:complete})}\\
            &\to \neg (\llbracket P' \rrbracket_{M} \sqsubseteq \llbracket P \rrbracket_{M}). \\
            &\to \exists E' \in \llbracket P' \rrbracket_{M} . \nexists E \in \llbracket P \rrbracket_{M} \ \text{s.t.} \ (E' \sim E). \tag*{(Def~\ref{def:safe-transf})} \\ 
            &\to E \in I\langle P \rangle_{M}. \tag*{($\langle P' \rangle \sqsubseteq \langle P \rangle$)}\\ 
            &\to E \in I\langle P \rangle_{B}. \ \tag*{(Def~\ref{def:weak-mem-model})}\\ 
            &\to E' \in I\langle P' \rangle_{B}.\ \tag*{(Def~\ref{def:safe-transf},\ref{def:weak-mem-model})} \\
            &\to \neg c_{B \setminus M}(E'). &
        \end{flalign*}
        Which means $\forall tr \ . \ \neg c_{B \setminus M}(E') \implies \neg \psf{B}{tr}{P}$ is false.
        Hence, proved. 
    \end{proof}

    Theorem~\ref{thm:complete-gen} helps us prove $M$ Complete w.r.t. $B$ by simply identifying those transformation-effects which need to be proven unsafe under concerned base model $B$. 
    \note{
    On a high level, we do this by identifying another $E'_{t} \in \llbracket P' \rrbracket_{B}$ which has no comparable $E_{t} \in \llbracket P \rrbracket_{B}$. 
    In order to derive such an execution, we manipulate specific $\rf$ relations of an inconsistent candidate execution.  
    The idea is that for a certain set of inconsistent $E'$ (at least under the memory models we consider), we can identify reads which, on the removal of their associated $\rf$ relation, give us a consistent execution instead (albeit not well-formed). 
    Fig~\ref{fig:crucial} shows the example of an inconsistent execution $E$ in the middle for a memory model with constraint $(\po \cup \rf)^{+} \ \text{acyclic}$. 
    We can see how removing the $\rf$ relation with $b=y$ or $a=x$  gives us a resultant executions $E1$ and $E2$ both of which are consistent under the same model.
    }
    \begin{figure*}[htbp]
        \centering
        \includegraphics[scale=0.6]{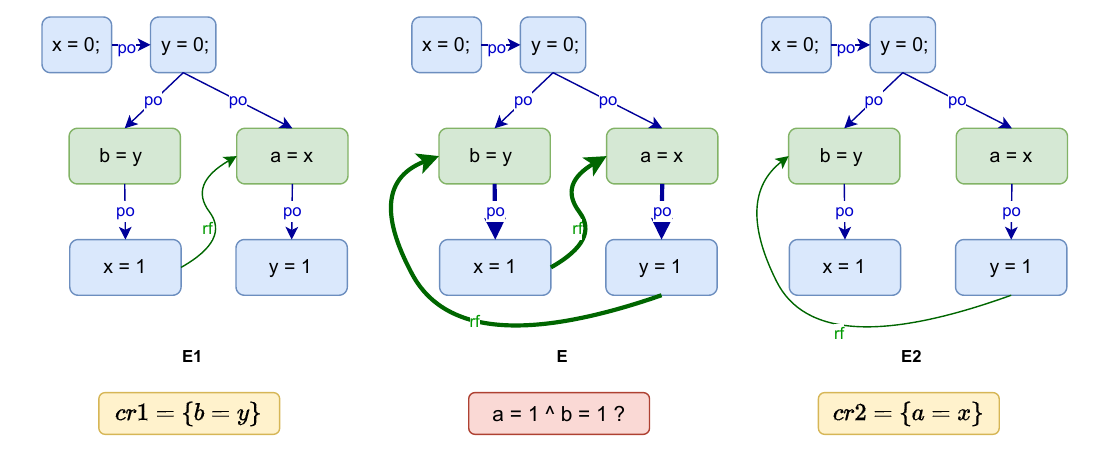}
        \caption{Removing the $\rf$ relation with $b=y$ (left) or with $a=x$ (right) from the inconsistent execution $E$ gives us consistent executions $E1$ and $E2$ respectively under memory model $(\po \cup \rf)^{+}$ acyclic.}
        \label{fig:crucial}
    \end{figure*}
    
    Formally, we denote such reads as \textit{crucial} w.r.t. a memory model for a candidate execution, and a collection of such reads as a crucial set.  
    \note{
    \begin{definition}
        \label{def:cr-rf}
        Given a execution $E$ inconsistent under memory model $M$, a set $cr \subseteq r(p(E))$ is a \emph{crucial set} if the execution $E' \sim E$ such that 
        \begin{tasks}(2)
            \task $\rf(E') = \rf(E) \setminus \ \{ [a];\rf;[b] \ | \ b \in cr \}$. 
            \task $\mo(E') = \mo(E)$.
        \end{tasks} 
        is consistent under $M$.
    \end{definition}
    Recalling example in Fig~\ref{fig:crucial}, both $cr1$ and $cr2$ are crucial sets.
    }

    \begin{remark}
        A crucial set may not always exist; the inconsistent execution in Fig~\ref{ex-mem-model} has no crucial set for the concerned memory model.
    \end{remark}
    
    \note{
    We now link crucial sets to the constraints of the memory model, which we recall will be in the form of irreflexivity constraints on binary relations.
    Violation of any such constraint would imply the existence of a cycle in the graph, like that in Fig~\ref{fig:crucial}.
    Therefore, the crucial reads would be those whose $\rf$ relations enable such a cycle to exist.  
    \begin{definition}
        \label{def:crucial-refl-comp}
        We define $cra$ as a function that takes in a binary relation $s$ between memory events in $E$ and returns a set of reads $r$ such that without each of its associated $\rf$, the relation is irreflexive.
        \begin{align*}
            \forall r \in cra(s) \ . \ r \notin st(p(E)) \implies s \ \text{irreflexive}
        \end{align*}
    \end{definition}
    }

    \footnotetext{
        The edges of the cycle will depend on the binary relations used to establish them. 
        These will be as per the memory model.
    }

    \note{
    With $cra$, we can construct crucial sets for each $E'$ inconsistent under $B$, choosing at most one read from each cyclic path formed which violate the constraints of $B$.
    Recall the purpose of identifying such a crucial set is to derive another $E'_{t}$ such that $c_{B}(E'_{t})$.
    We can do this if it is possible to reassign the $\rf$ relations of the reads in the crucial set to give us our desired $E'_{t}$.
    For the models we consider, we show this is possible by first removing all the $\rf$ relations from one of its crucial sets, followed by incrementally adding \emph{consistent} $\rf$ relations one by one.
    We term this property as \emph{piecewise consistent} that a memory model satisfies for any execution $E$.    
    Formally, 
    \begin{definition}
        \label{def:piecewise-cons}
        A memory model $M$ is \emph{piecewise consistent} ($\pwc{M}$), if for any execution $E$ such that 
        \begin{tasks}(2)
            \task $\neg \wf{E}$.
            \task $c_{M}(E)$.  
            \task $\rf^{-1}$ functional.
            \task $\mo$ total order.
            \task $\forall loc \ . \ max(loc) \neq \phi$.
        \end{tasks}
        we have a candidate execution $E_{t}$ such that 
        \begin{tasks}(2)
            \task $c_{M}(E_{t})$.
            \task $p(E) = p(E_{t})$.
            \task $\mo(E) = \mo(E_{t})$.
            \task $\rf(E) \subset \rf(E_{t})$.
        \end{tasks}
    \end{definition}
    }

    We show in Sec~\ref{subsec:scrr-conc} how the above ideas are utilized to prove Complete for a concrete memory model w.r.t. $SC$.

    \section{Towards Concrete Models}
    \label{sec:concrete}
    We now apply our formalism for concrete axiomatic models.
    First, we specify the axiomatic model of \textit{SC} in the form of irreflexivity constraints (Sec~\ref{subsec:sc-conc}).
    This is followed by showcasing with our formalism why \textit{TSO} is not Complete w.r.t. \textit{SC} (Sec~\ref{subsec:tso-conc}). 
    Next, we derive $\textit{SC}_\textit{RR}$ as a model that allows independent read-read reordering over \textit{SC}.
    We show that $\textit{SC}_\textit{RR}$ is Weak, Sound and Complete for transformation-effects not involving write-elimination (Sec~\ref{subsec:scrr-conc}).
    Then, we show that $\textit{SC}_\textit{RR}$ remains Weak, Sound and Complete (barring write-eliminations) w.r.t. \textit{SC} on inclusion of \textit{fences} ($frr$) and \textit{read-modify-write} ($rmw$) events in the language (Sec~\ref{subsec:scfrr-conc}).
    Lastly, we conclude by showing why it cannot be proven Complete for effects involving write-elimination (Sec~\ref{subsec:scrr-welim}).

    \paragraph*{Additional relations}
    First, let \textit{read-from-internal} ($\rfi$) represent the subset of $\rf$ such that both the write and read are of the same thread,let \textit{reads-from-external} ($\rfe$) be the rest.
    Second, let \textit{happens-before} ($\hb$) be the transitive closure of program order and reads-from relations $(\po \cup \rf)^{+}$.
    Third, let \textit{memory-order-loc} ($\mo_{|loc}$) represent the memory order between writes to same memory.
    Finally, let \textit{reads-before} ($\rb$) represent the sequential composition $\rf^{-1};\mo_{|loc}$, which relates reads to writes which have taken place before it to same memory in an execution.  formally

\subsection{Sequential Consistency (SC)}

    \label{subsec:sc-conc}

    We adopt the axiomatic format of \textit{SC} as in \cite{LahavV}, which consists of the following set of consistency rules for any execution $E$
    \begin{tasks}(3)
        \task $\mo$ strict total order.
        \task $\hb$ irreflexive.
        \task $\mo;\hb$ irreflexive. 
        \task $\rb;\hb$ irreflexive.
        \task $\rb;\mo;\hb$ irreflexive.
    \end{tasks}

    Fig~\ref{fig:sc-wr-axiomatic} show the two executions of store buffering discussed earlier in Fig~\ref{SB}. 
    The outcome to the left has no rule of \textit{SC} is violated, thus deeming it consistent under \textit{SC}.
    Whereas the outcome to the right is inconsistent, the cycle $[a=x];\rb;[x=1];\mo;[y=1];\po;[a=x]$ violates Rule (e).
    \begin{figure*}[htbp]
        \centering
        \includegraphics[scale=0.6]{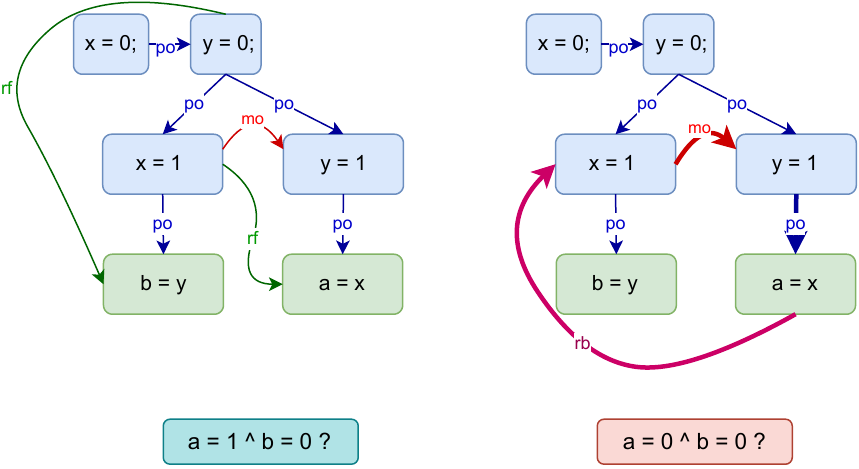}
        \caption{Two outcomes of store buffering; one consistent (left) and another inconsistent (right) under \textit{SC}.}
        \label{fig:sc-wr-axiomatic}
    \end{figure*}

    \note{
    We use Def~\ref{def:crucial-refl-comp} to correlate \textit{SC} with the notion of \textit{crucial reads}.
    Recall that the cyclic path in Fig~\ref{fig:sc-wr-axiomatic} right results in the sequential composition $\rb;\mo;\hb$ to be no longer irreflexive.
    Specifically, $[a=x];\rb;[x=1];\mo;[y=1];\po;[a=x]$ is non-empty.
    \begin{restatable}{lemma}{lemsccra}
        \label{lem:sc-cra}
        For a given execution $E$, we have 
        \begin{align*}
            &cra(\mo) = \phi \\
            &cra(\hb) = \{r \ | \ ([e];\rfe;[r];\hb;[e] \neq \phi) \vee ([e];\rfi;[r];\po;[e] \neq \phi) \} \\ 
            &cra(\mo;\hb) = \{r \ | \ [e];\mo;\hb^{?};\rfe;[r];\hb;[e] \neq \phi \} \\
            &cra(\rb;\hb) = \{r \ | \  ([e];\rb;\hb^{?};\rfe;[r];\hb;[e] \neq \phi) \vee ([r];\rb;\hb;[r] \neq \phi) \} \\
            &cra(\rb;\mo;\hb) = \{r \ | \ ([e];\rb;\mo;\hb^{?};\rfe;[r];\hb;[e] \neq \phi) \vee ([r];\rb;\mo;\hb;[r] \neq \phi) \}
        \end{align*}  
    \end{restatable}
    }
    \begin{proof}(sketch)
        Consider the case for $\hb$ which is not irreflexive, i.e. represented by the relation $\hb$ cycle.
        \begin{align*}
            &[e];\hb; \ \text{cycle} \\  
            &\to [e];\rfe;[r];\hb \vee [e];\rfi;[r];\po \\ 
            &\to [e];\rfe;[r];\po;\hb;\po \vee [e];\rfi;[r];\po \ \text{cycle}. 
        \end{align*}
        We can observe that for both cases, without the $\rf$ relation with $r$ and $e$, the cycle will cease to exist. 
        The sequential composition no longer exists, meaning the aforementioned $\hb$ cycle no longer exists \footnotemark.

        The other cases follow a similar line of argument.
    \end{proof}
    
    \footnotetext{
        For $\hb$ to be irreflexive, all such $\hb$ cycles must disappear or no longer exist.
    }
    
    We refer readers to Appendix~\ref{appendix:B} for the proofs that the consistency rules of \textit{SC} are not redundant, and that \textit{SC} is piecewise consistent.

\subsection{Revisiting \textit{TSO} vs \textit{SC}}

    \label{subsec:tso-conc}

    We revisit our counter example on inlining which showcased TSO is not Complete w.r.t. \textit{SC}.
    We do not derive \textit{TSO} from \textit{SC}, rather we use a similar existing formulation of the model in \cite{LahavV}.
    The fragment of \textit{TSO} without fences $F$ and \textit{read-modify-writes/lock}s have the following consistency rules for any execution $E$. 
    \begin{tasks}(2)
        \task $\mo$ strict total order.
        \task $\hb$ irreflexive.
        \task $\mo;\hb$ irreflexive.
        \task $\rb;\hb$ irreflexive.
        \task $\rb;\mo;\rfe;\po$ irreflexive.
    \end{tasks}
    The rules (a), (b), (c), (d) are the same as that for \textit{SC}, whereas rule (e) is a subset of that in \textit{SC}; which specifically allow the behaviors exhibited by FIFO store buffering of TSO.
    As an example, Fig~\ref{fig:tso-wr-axiomatic} shows both the executions of store buffering shown for \textit{SC} in Fig~\ref{fig:sc-wr-axiomatic}. 
    Notice that the outcome on the right is also allowed; no rule of \textit{TSO} is violated.
    \begin{figure*}[htbp]
        \centering
        \includegraphics[scale=0.6]{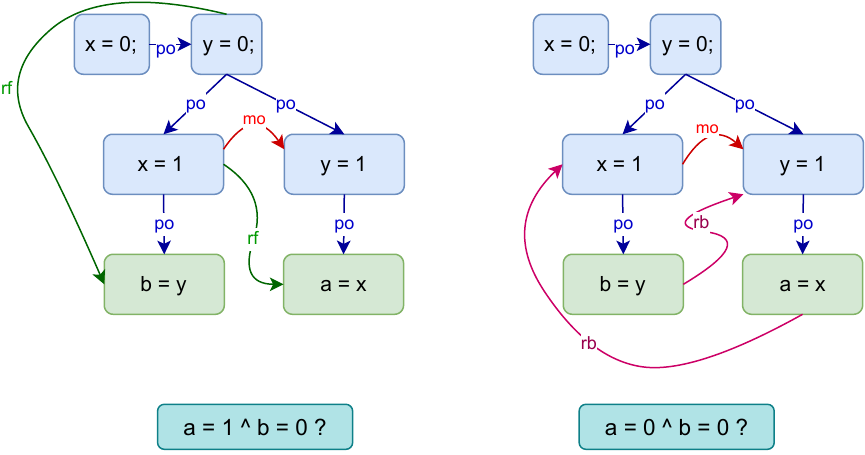}
        \caption{Two outcomes of store buffering, both consistent under \textit{TSO}.}
        \label{fig:tso-wr-axiomatic}
    \end{figure*}
    
    We show why \textit{TSO} is not Complete w.r.t. \textit{SC} using our formalism. 
    We make use of Theorem~\ref{thm:complete-gen} as follows:
    \begin{align*}
        &B = \textit{SC}. & M = \textit{TSO}. \\
        &\textit{weak}\langle \textit{SC}, \textit{TSO} \rangle & B \setminus M = \rb;\mo;\po \ \text{irreflexive}. 
    \end{align*}
 
    To prove Complete, we need to show all transformation-effects quantified by Theorem~\ref{thm:complete-gen} as unsafe under \textit{SC}.
    Even if one of these transformation-effects are safe instead, then by Def~\ref{def:complete}, \textit{TSO} is not Complete w.r.t. \textit{SC}.  
    Thread-inlining qualifies to be such an effect; Fig~\ref{fig:sc-tso-incomplete} shows two executions $E$ (left) and $E'$ as per our requirement.
    We can see that for $E'$ consistent under \textit{TSO}, there only exists inconsistent execution $E$ comparable from the original, i.e. $E \sim E'$. 
    We can see that $E$ (left) is inconsistent both under \textit{TSO} and \textit{SC} due to $[b=x];\rb;[x=1];\mo;[y=1];\rfe;[y];\po$ cycle, violating rule (e) of both \textit{SC} and \textit{TSO}.
    However, for $E'$ (right), note that it is only inconsistent under \textit{SC} due to $[b=x];\rb;[x=1];\mo;[y=1];\po$ cycle, which does not violate any constraint of \textit{TSO}. 
    Since such an inlining-effect is safe (and also sound \cite{MoiseenkoP}) under \textit{SC}, we have $\neg \text{comp}\langle \textit{TSO}, \textit{SC} \rangle$. 
    \begin{figure*}[htbp]
        \centering
        \includegraphics[scale=0.6]{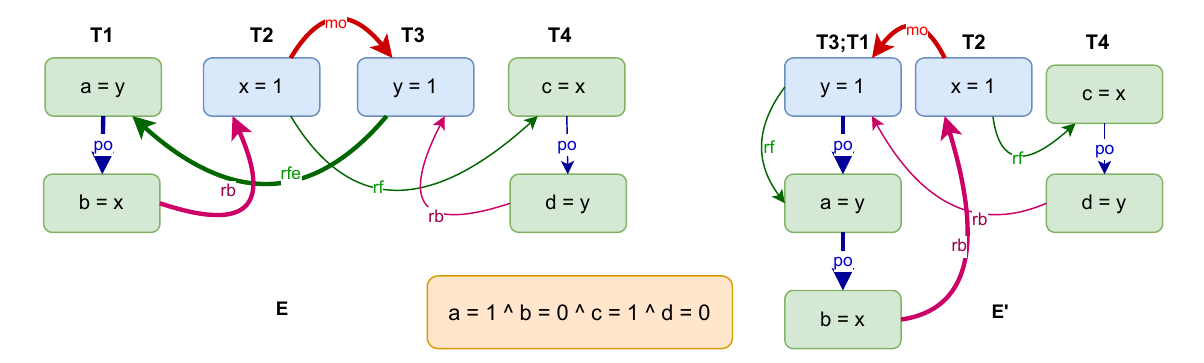}
        \caption{For the given outcome (yellow box), $E$ is inconsistent under both \textit{SC} and \textit{TSO}, but $E'$ is only inconsistent under \textit{SC}.}
        \label{fig:sc-tso-incomplete}
    \end{figure*}

    \subsection{SC + Reordering of Independent Read-Read (\textbf{$\textit{SC}_\textit{RR}$})}

    \label{subsec:scrr-conc}

    With Theorem~\ref{thm:complete-gen}, we identified transformation-effects using which we showed $\neg \comp{TSO}{SC}$. 
    However, the main question remains; is it practically possible to prove a given consistency model Complete  
    w.r.t. another ?
    If not for all transformation-effects, can we instead identify a significant set for which we can ?
    In this section, we show that it is possible to do so. 
    We first derive an axiomatic model from \textit{SC} with the intent on allowing read-read reordering.
    We then show the resultant model is Weak, Sound and Complete, the latter of which is restricted to effects not involving write elimination.  
    \subsubsection{Transformation-effect and Proposed Model}

        We first define reordering of adjacent independent reads as an effect.
        Let $\po_{imm}$ represent the program order between two adjacent events, i.e. no other memory event exists between them in program order ($[a];\po_{imm};[b] \implies \nexists c \ .\ [a];\po;[c];\po;[b]$). 
        \begin{definition}
            \label{def:rr-ind}
            Consider a pre-trace $P$ having two reads $r_{x}, r_{y}$ such that $mem(r_{x}) \neq mem(r_{y})$ and $[r_x];\po_{imm};[r_y]$.
            Then, \emph{reordering of adjacent independent reads} $r_{x}, r_{y}$ is a transformation-effect $P \mapsto_{rr} P'$ such that  
            \begin{tasks}(2)
                \task $r(P') = r(P)$. 
                \task $w(P') = w(P)$.
                \task $\po^{-} = [r_x];\po_{imm};[r_y]$.
                \task $\po^{+} = [r_y];\po_{imm};[r_x]$.
            \end{tasks}
        \end{definition}

        Transformation-effect in Def~\ref{def:rr-ind} can be shown unsafe under \textit{SC} using the example in Figure~\ref{sc-mp}.
        The outcome (yellow box) in the original pre-trace is not possible under \textit{SC} (left) as the candidate execution justifying it has $[b=x];\rb;[x=1];\po;[y=1];\rfe;[a=y];\po$ cycle, thereby violating rule (d) of \textit{SC}.
        However, the outcome is allowed under \textit{SC} as a result of reordering the two adjacent reads $[a=y];\po_{imm};[b=x]$ (right), thereby removing any cyclic path that violate rules of \textit{SC}.
        \begin{figure*}[htbp]
            \centering
            \includegraphics*[scale=0.6]{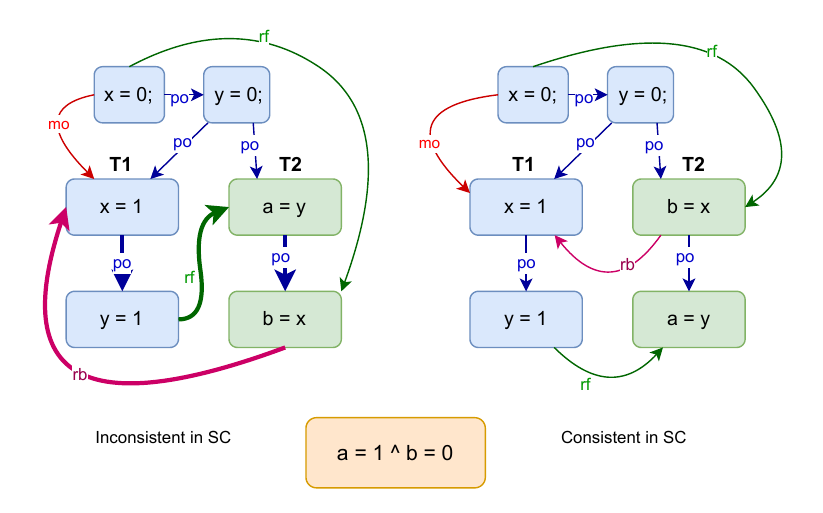}
            \caption{The forbidden outcome under \textit{SC} (yellow box) for the original pre-trace (left) allowed after reordering the two reads in $T2$ (right).}
            \label{sc-mp}
        \end{figure*}
    
        To identify the axiom that prevents $rr$, note that the transformation-effect in Fig~\ref{sc-mp} was unsafe primarily due to a subset of $\rb;\hb$ that involved the $\po_{imm}$ between the two reads ($\rb;\hb;\rfe;[a=y];\po_{imm};[b=x]$).
        More generally, we note that in any $\rb;\mo^{?};\hb$ cycle, the $\hb$ relation between any two independent reads like that above plays a role in disallowing any possible adjacent read-read reordering. 
        Thus, the axiomatic model $\textit{SC}_\textit{RR}$ that allows reordering of adjacent independent reads is the following constraint on any execution $E$.  
        \begin{align*}
            SC_{RR} = SC \setminus (\rb;\mo^{?};\hb^{?};\rfe;[r_{loc}(p(E))];\hb;[r_{!loc}(p(E))] \ \text{irreflexive}).
        \end{align*}
        where $\mo^{?}, \hb^{?}$ here denote they are optional in the sequential composition, and $loc$ represents any memory location.
        We refer to the removed axiom as $a_{rr}$ for brevity.
        To be specific, $\textit{SC}_\textit{RR}$ is the derived model having the following consistency rules for any execution $E$.
        \begin{tasks}(3)
            \task $\mo$ strict and total.
            \task $\hb$ irreflexive.
            \task $\mo;\hb$ irreflexive.
            \task $\rb;\hb \setminus a_{rr}$ irreflexive.
            \task $\rb;\mo;\hb \setminus a_{rr}$ irreflexive.
        \end{tasks}
        Rules (a), (b), (c) are the same as that of \textit{SC}, whereas rules (d) and (e) are ones which require the removal of $a_{rr}$.

    \subsubsection{Sound and Complete}

        We first show $\textit{SC}_\textit{RR}$ is Sound.
        \begin{restatable}{theorem}{thmrrsound}
            \label{thm:sc-rr-sound}
            For transformation-effect $rr$ in Def~\ref{def:rr-ind}, we have $sound(SC_{RR}, rr)$.
        \end{restatable}

        \begin{proof}(sketch)

            We prove this by contradiction. 
            \begin{itemize}
                \item Suppose $rr$ is not Sound.
                \item This implies there exists some pre-trace $P$, for which the transformation-effect $P \mapsto_{rr} P'$ is unsafe ($\neg \psf{M}{rr}{P}$). 
                \item From Def~\ref{def:safe-transf}, this would mean there exists some consistent execution $E' \in \llbracket P' \rrbracket_{SC_{RR}}$ for which there does not exist any comparable execution $E \in \llbracket P \rrbracket_{SC_{RR}}$.
                \item Since for all well-formed executions, we assume $\langle P' \rangle \sqsubseteq \langle P \rangle$, we can infer that every such comparable $E$ is inconsistent under $\textit{SC}_\textit{RR}$.
                \item Considering $r_{x}, r_{y}$ to be two adjacent reads which have been reordered, we can infer that $[r_{x}];\po_{imm};[r_{y}]$ is responsible in $E$ being inconsistent. 
                \item The $\po_{imm}$ relation between the reads plays a role in violating some rule of $\textit{SC}_\textit{RR}$.
                \item Whereas on reordering them, $E'$ is consistent instead.
            \end{itemize}

            We show that this is never true; for any such $E$, we must also have $E'$ inconsistent under $\textit{SC}_\textit{RR}$.   
            As a sample, let us say the Rule (b), i.e. $\hb \ \text{irreflexive}$ does not hold in $E$, i.e. $\hb$ forms a cycle. 
            \begin{align*}
                &\to \hb \ \text{cycle} \\
                &\to \hb;[r_{x}];\po_{imm};[r_{y}];\hb  \ \\ 
                &\to \hb;[r_{x}];\po_{imm};[r_{y}];\po; \ \\ 
                &\to \hb;\po;[r_{x}];\po_{imm};[r_{y}];\po; \ \vee \ \hb;\rfe;[r_{x}];\po_{imm};[r_{y}];\po;. \\
                &\to \hb;\po;[r_{y}] \ \vee \ \hb;\rfe;[r_{x}];\po.  
            \end{align*}
            We can see that $\hb$ forms a cycle independent of the $\po_{imm}$ between $r_{x}$ and $r_{y}$, allowing $E'$ to also be inconsistent due to Rule (b). 
            Showing the same for all possible violated rules of $\textit{SC}_\textit{RR}$ helps us conclude we indeed have $sound(SC_{RR}, rr)$.
            The rest of the cases are in Appendix~\ref{appendix:C}.
        \end{proof}


        Next, we show that $\textit{SC}_\textit{RR}$ is Complete w.r.t. \textit{SC} for a significant set of transformation-effects.
        We use the result from Theorem~\ref{thm:complete-gen} for our purpose as follows:
        \begin{align*}
            &B = \textit{SC}. & M = \textit{SC}_{\textit{RR}}. \\
            &\textit{weak}\langle \textit{SC}, \textit{SC}_{\textit{RR}} \rangle & B \setminus M = a_{rr} 
        \end{align*}
 
        Recall from Sec~\ref{subsec:meth-comp} that our objective is to derive another $E'$ of the transformed pre-trace which is consistent under \textit{SC}, but has no consistent execution $E$ from the original pre-trace.
        For this, we make use of $cra$ (Def~\ref{def:crucial-refl-comp}) as follows. 
        We first note the relation between compositions violating $a_{rr}$ and 4 other sequential compositions which may not be irreflexive. 
        \note{
        \begin{proposition}
            \label{prop:sc-rr:cra}
            Let $rr$ be the composition $\rb;\mo^{?};\hb^{?};\rfe;[r_{x}];\hb;[r_{y \neq x}]$ which violates constraint $a_{rr}$, i.e. it forms a cycle.
            Then, we have the following 
            \begin{tasks}(2)
                \task $cra(rr) \nsubseteq cra(\rb;\mo^{?};\po;[r_{x}])$. 
                \task $cra(rr) \nsubseteq cra(\rb;\mo^{?};\po;[r_{y}])$.
                \task $cra(rr) \nsubseteq cra(\rb;\rfe;[r'_{x}];\po;[r_{x}])$.
                \task $cra(rr) \nsubseteq cra(\rb;\rfe;[r'_{y}];\po;[r_{y}])$. 
                \task $cra(rr) \nsubseteq cra(\mo;\rfe;[r_{x}];\po)$.
                \task $cra(rr) \nsubseteq cra(\mo;\rfe;[r_{y}];\po)$.
                \task $cra(rr) \nsubseteq cra(\rfi;[r_{x}];\po)$.
                \task $cra(rr) \nsubseteq cra(\rfi;[r_{y}];\po)$.
            \end{tasks}        
        \end{proposition}
        }
    
        We then show that for our purpose, we have at least one of the above compositions is true in $E$, as described in Theorem~\ref{thm:complete-gen}; corelating it with violating constraints common to both \textit{SC} and $\textit{SC}_\textit{RR}$.
        \note{
        \begin{restatable}{lemma}{lemscrrcompgen}
            \label{lem:sc-rr:comp:general}
            Consider transformation-effects $P \mapsto_{tr} P'$ defined by Theorem~\ref{thm:complete-gen} not involving write-elimination ($st^- \cap w(P) = \phi$).
            Consider also, from Theorem~\ref{thm:complete-gen}, the corresponding execution $E' \in \langle P' \rangle$ which is consistent under $\textit{SC}_\textit{RR}$ but inconsistent under \textit{SC}. 
            Then, at least one of the following is true for any $E \in \langle P \rangle$ where $E \sim E'$.  
            \begin{tasks}(2)
                \task $\rfi;\po$ cycle. 
                \task $\mo;\po$ cycle. 
                \task $\mo;\rfe;\po$ cycle.
                \task $\rb;\mo^{?};\po$ cycle.
                \task $\rb;\rfe;\po$ cycle. 
            \end{tasks}
        \end{restatable}
        }
    
        \begin{proof}(sketch)

            By Theorem~\ref{thm:complete-gen}, we know that any $E$ such that $E \sim E'$ is inconsistent under $\textit{SC}_\textit{RR}$.
            Going case wise for each possible violated constraint of $\textit{SC}_\textit{RR}$, we show that it reduces to one of the above mentioned compositions.
            We refer the readers to the Appendix~\ref{appendix:C} for the detailed proof.

        \end{proof}

        Finally, using the above, we can construct and compare crucial sets for both $E'$ and $E$, allowing us to show $\textit{SC}_\textit{RR}$ is complete w.r.t. \textit{SC} for all effects that do not involve write elimination.
        \begin{restatable}{theorem}{thmscrrcompnowelim}
            \label{thm:sc-rr:comp-no-w-elim}
            For all transformation-effects not involving write elimination ($st^{-} \cap w(P) = \phi$), $\text{comp}\langle SC_{RR}, SC \rangle$.
        \end{restatable}

        \begin{proof}(sketch)
            Note that each relation that forms a cycle as specified in Lemma~\ref{lem:sc-rr:comp:general} also violates one of the constraints of \textit{SC}.
            From Prop~\ref{prop:sc-rr:cra}, we can then derive a crucial set w.r.t. \textit{SC} for $E'$ and can show it isn't the same for $E$. 
            Using $pwc(SC)$, we can derive a pair of candidate execution pairs $E_{t} \sim E'_{t}$ such that $E'_{t} \in \llbracket P' \rrbracket_{SC}$ and $E_{t} \in I\langle P \rangle_{SC}$.
            Since no write-elimination occurs, by Def~\ref{def:transf-safe} the transformation-effect $tr$ is unsafe under \textit{SC}.
            Finally, by Theorem~\ref{thm:complete-gen}, we have $\comp{SC_{RR}}{SC}$.
            We refer the readers to Appendix~\ref{appendix:C} for the detailed formal proof.
        \end{proof}

    \subsection{Read-Modify-Writes and Fence Instructions}

    \label{subsec:scfrr-conc}

    So far, we have only addressed programs that have simple memory read/write events.
    In practice, most concurrent algorithms also rely on read-modify-write (rmw) instructions; memory events that perform a read followed by a write in one atomic step \cite{LawOfOrder} (eg: compare-and-swap, fetch-and-add, etc.)
    To add, the model $\textit{SC}_\textit{RR}$, when desired, does not allow us to forbid any form of read-read reordering.
    Both these requirements can be fulfilled respectively by adding a read-modify-writes and a fence event to the language.    
    
    Let $rmw(a,x,v)$ represent a read from shared memory $x$ into thread-local variable $a$, followed by a write of value $v$ to it and $rmw(P)$ return the read-modify-write events of pre-trace $P$.
    Let $f_{rr}$ represent the fence event and $frr(P)$ return the fence events of pre-trace $P$.

    \subsubsection{Including $rmw$ and $f_{rr}$ Rules}
        
        First, we consider $rmw$ to be both reads and writes - $rmw(P) = r(P) \cap w(P)$.
        This ensures all existing rules of \textit{SC} and $\textit{SC}_\textit{RR}$ also apply for any read-modify-write event. 
        Second, we add an additional rule for $rmw$ \textit{atomicity}: no memory event must appear to occur between the read and write of this event in any execution.
        From \cite{LahavV}, this requirement translates to forbidding outcomes like those in the left in Fig~\ref{fig:rmw-frr}.
        Here, we note that $a=0$ implies the read has taken place before the write $y=2$.
        However, the final value of $y$ being $1$ indicates that $y=2$ took place before $y=1$.
        Such an execution has $[rmw(a,y,1)];\rb;[y=2];\mo$ cycle. 
        Thus, for \textit{SC} and all models derived from it, the rule of atomicity would be 
        \begin{align*}
            \rb;\mo \ \text{irreflexive}.
        \end{align*}
       
        In order to disallow reordering two independent reads when desired under $\textit{SC}_\textit{RR}$, our idea is to place an $rmw$ or fence $f_{rr}$ event in between them.
        \begin{figure*}[htbp]
            \centering
            \includegraphics[scale=0.6]{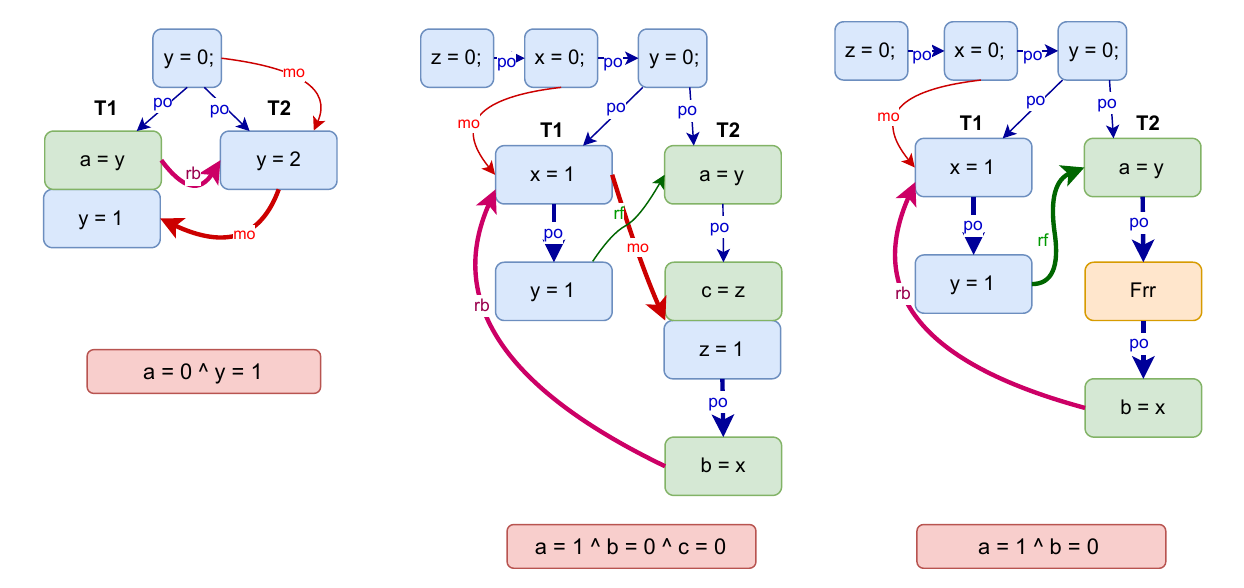}
            \caption{The role of read-modify-writes and fences in forbidding outcomes violating atomicity (left) and preventing reordering of independent reads (middle, right).}
            \label{fig:rmw-frr}
        \end{figure*}
        Concretely, if there is an $rmw$ event between two read events in program order, we do not wish to permit their reordering. 
        Since $\mo$ also involves $rmw$ events, we need to add the following rule:
        \begin{align*}
            \rb;\mo^{?};[rmw];\po;[r] \ \text{irreflexive}.
        \end{align*}
        This translates to disallowing the outcomes like those in the middle in Fig~\ref{fig:rmw-frr}.
        The corresponding candidate execution has $[a=y];\po;[rmw(c,z,1)];\po;[b=x]$, and the outcome $a=1 \ \wedge \ b=0 \ \wedge \ c=0$ is inconsistent under $\textit{SC}_\textit{RR}$ as $[b=x];\rb;[x=1];\mo;[rmw(c,z,1)];\po;[b=x]$ is non-empty.  
           
        Similarly, if a fence event $f_{rr}$ exists between two reads $r_{y}$, $r_{x}$ in program order, reordering the two reads should not be allowed.
        Using the same idea for $rmw$, the rule to add is similar to $a_{rr}$, except replacing $rmw$ with a fence instead, giving us the following consistency rule given any execution $E$
        \begin{align*}
            \rb;\mo;\rfe;[r_{loc}(p(E))];\po;[frr(p(E))];\po;[r_{!loc}(p(E))] \ \text{irreflexive}.
        \end{align*}
        We refer to the fence rule as $a_{frr}$ for brevity.
        This translates to disallowing the outcomes like those in the right in Fig~\ref{fig:rmw-frr}.
        The corresponding candidate execution has $[a=y];\po;[f_{rr}];\po;[b=x]$, and the outcome $a=1 \ \wedge \ b=0$ is inconsistent under $\textit{SC}_\textit{RR}$ as $[b=x];\rb;[x=1];\hb;[a=y];\po;[f_{rr}];\po;[b=x]$ is non-empty.  

        Our resultant model of $\textit{SC}_\textit{RR}$, with the inclusion of the above rules is 
        \begin{tasks}(3)
            \task $\mo$ strict and total.
            \task $\hb$ irreflexive.
            \task $\mo;\hb$ irreflexive.
            \task $\rb;\hb \setminus a_{rr}$ irreflexive.
            \task $\rb;\mo;\hb \setminus a_{rr}$irreflexive.
            \task $\rb;\mo$ irreflexive.
            \task $a_{frr}$.
        \end{tasks}
        The rule for $rmw$ is already a part of $\textit{SC}_\textit{RR}$; $a_{rr}$ specifically involves only plain reads.
        Therefore, we do not need to add it.
        Note also that $a_{frr}$ will only be required for $\textit{SC}_\textit{RR}$; the fence instruction is equivalent to a no-op under \textit{SC}, making the rule redundant for it. 
        
    \subsubsection{Weak, Sound, Complete?}

        The addition of read-modify-writes and read-read-fences along with their respective consistency rules for \textit{SC} and $\textit{SC}_\textit{RR}$ do not change our results.
        $\textit{SC}_\textit{RR}$ is still Weaker than \textit{SC}; the axiom $a_{frr}$ is redundant for \textit{SC}.
        For Sound, we need to address two more cases based on the added axioms of atomicity and fences.
        \begin{theorem}
            \label{lem:sc-frr-sound}
            For the $\textit{SC}_\textit{RR}$ model with the constraints (f) and (g) and transformation-effect $rr$ as in Def~\ref{def:rr-ind}, we have $sound(SC_{RR}, rr)$.
        \end{theorem}

        \begin{proof}
            
            Recall that in order to prove the effect $P \mapsto_{rr} P'$ is sound, we need to show that every execution $E \in \langle P \rangle$ which is inconsistent due to reads which are reordered, has a comparable execution $E' \in \langle P' \rangle$ which is also inconsistent.   
            In addition to the existing proof cases of Theorem~\ref{thm:sc-rr-sound}, two additional cases related to $rmw$ and $f_{rr}$ are to be addressed.
            \begin{itemize}
                \item Case 1: $\rb;\mo$ cycle
                    There is no $\po$ involved in the composition of this relation.
                    So none of the $\po$ edges matter.
                \item Case 2: $a_{frr}$ violated. This means there exists reads $r_{x}$, $r_{y}$ and a fence $f_{rr}$ such that 
                    \begin{align*}
                        &\rb;\mo;\rfe;[r_{y}];\po;[f_{rr}];\po;[r_{x}]  \ \text{cycle}. \\ 
                        &\to \rb;\mo;\rfe;[r_{y}];\po_{imm};[r_{z}];\po;[f_{rr}];\po;[r_{x}] \ \vee \\ 
                        &\rb;\mo;\rfe;[r_{y}];\po;[f_{rr}];\po;[r_{z}];\po_{imm};[r_{x}] \ \text{cycle}.
                    \end{align*}
                    For both cases, we note that $[r_{y}];\po_{imm};[r_{z}]$ as well as $[r_{z}];\po_{imm};[r_{x}]$ are not essential to violate $a_{frr}$.
            \end{itemize}
            Thus, by contradiction, we have $sound(SC_{RR}, rr)$.
        \end{proof}

        To show the result of Complete remains unchanged involves a little more work.
        First, we note relations between $cra(rr)$ and some additional compositions which form a cycle.
        \begin{proposition}
            \label{prop:sc-rr:rmw-frr}
            Let $rr$ be the composition $\rb;\mo^{?};\hb^{?};\rfe;[r_{x}];\hb;[r_{y \neq x}]$ which violates constraint $a_{rr}$, i.e. it represents a cyclic path.
            Then, we have the following 
            \begin{align*}
                &cra(rr) \nsubseteq cra(\rb;\mo) \\
                &cra(rr) \nsubseteq cra(\mo;\rf) 
            \end{align*}
        \end{proposition}
        
        Second, we have three more compositions that may not be irreflexive in any execution $E$ as specified by Theorem~\ref{thm:complete-gen}.
        \begin{restatable}{lemma}{lemscrrcomprmwfence}
            \label{lem:sc-rr:comp:rmw-fence}
            Consider transformation-effects $P \mapsto_{tr} P'$ defined by Theorem~\ref{thm:complete-gen} not involving write-elimination ($st^- \cap w(P) = \phi$) and the corresponding execution $E' \in \langle P' \rangle$ consistent under $\textit{SC}_\textit{RR}$ but inconsistent under \textit{SC}. 
            Then, on addition of $f_{rr}$ and $rmw$ events and their consistency rules, the following may also be true for any $E \in \langle P \rangle$ where $E \sim E'$.
            \begin{tasks}(2)
                \task $\mo;\rf$ cycle.
                \task $\rb;\mo$ cycle.
                \task $\rb;\mo;\rfe;\po;[f_{rr}];\po$ cycle.
            \end{tasks}
        \end{restatable}
        Third, remember that proving $\textit{SC}_\textit{RR}$ Complete w.r.t. \textit{SC} does not involve transformation-effects with $f_{rr}$; they exist simply as no-ops.
        Accordingly, we identify the cases when a transformation-effect involves $f_{rr}$.
        \begin{restatable}{lemma}{lemscrrfrr}
            \label{lem:sc-rr:frr}
            For transformation-effects $P \mapsto_{tr} P'$ as defined in Theorem~\ref{thm:complete-gen} with no write elimination, and with corresponding pair $E \sim E'$ and $B = SC$, $M = SC_{RR}$, we have the following             
            \begin{align*}
                &a_{frr}(E) = \text{false} \ \implies \ ((f_{rr} \in st^{-}) \vee ([f_{rr}];\po;[r_{x}] \in \po^{-}) \vee ([r_{y}];\po;[f_{rr}] \in \po^{-})).
            \end{align*}
            where $r_{x}, r_{y} \notin st^{-}$.
        \end{restatable}

        Finally, we show that addition of $f_{rr}$ and $rmw$ events and their consistency rules (f) and (g) does not change the result of Theorem~\ref{thm:sc-rr:comp-no-w-elim}.        
        \begin{restatable}{theorem}{thmscrrcomprmwfrr}
            \label{thm:scrr:comp:rmw-frr-no-w-elim}
            Given \textit{SC} and $\textit{SC}_\textit{RR}$ with $rmw$ and $f_{rr}$ events, for all transformation-effects not involving write elimination, we have $\comp{SC_{RR}}{SC}$.
        \end{restatable}

        The proofs of Lemma~\ref{lem:sc-rr:comp:rmw-fence}, \ref{lem:sc-rr:frr} and Theorem~\ref{thm:scrr:comp:rmw-frr-no-w-elim} are in Appendix~\ref{appendix:D}.

    \subsection{Write Elimination-effects}

    \label{subsec:scrr-welim}

    We acknowledge that we have proven $\textit{SC}_\textit{RR}$ Complete w.r.t. \textit{SC} only for transformation-effects without write-elimination.
    The reason for this can be explained via a simple counter example in Fig~\ref{sc-rr-vr}. 
    Consider the original program on the left. 
    The compiler, on performing value range analysis, can declare the write $z=0$ as dead-code and remove it, resulting in the code to the right. 
    \begin{figure*}[htbp]
        \centering
        \includegraphics*[scale=0.6]{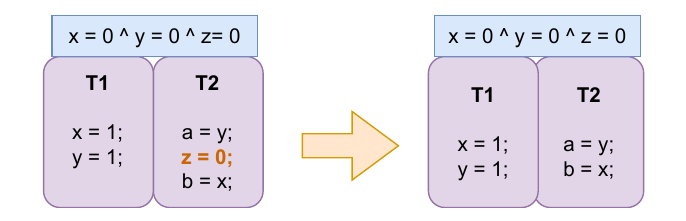}
        \caption{The write $z=0$ eliminated as dead-code on performing value range analysis.}
        \label{sc-rr-vr}
    \end{figure*}
    While such a transformation-effect is safe under \textit{SC} (we welcome to readers to check it themselves for all possible candidate executions), the same cannot be said for $\textit{SC}_\textit{RR}$.
    Observe the two similar candidate executions $E \sim E'$ in Fig~\ref{sc-rr-welim}, the left of the original program's pre-trace $P$ and the right for the transformed pre-trace $P'$. 
    We can see that the behavior (yellow box) justified by $E$ is inconsistent under $\textit{SC}_\textit{RR}$, the cycle $[b=x];\rb;[x=1];\mo;[z=0];\po;[b=x]$ in $E$, violates consistency rule (e).
    Whereas $E'$ is consistent under $\textit{SC}_\textit{RR}$ for $P'$.
    To finish, note that even on changing the $\mo$ between $y=1$ and $z=0$, the resultant candidate execution still remains inconsistent under $\textit{SC}_\textit{RR}$, due to the cycle $[z=0];\mo;[y=1];\rfe;[a=y];\po;[z=0]$ violating constraint (c).
    Thus, by Def~\ref{def:transf-safe}, this transformation-effect is unsafe under $\textit{SC}_\textit{RR}$.  
    \begin{figure*}[htbp]
        \centering
        \includegraphics*[scale=0.6]{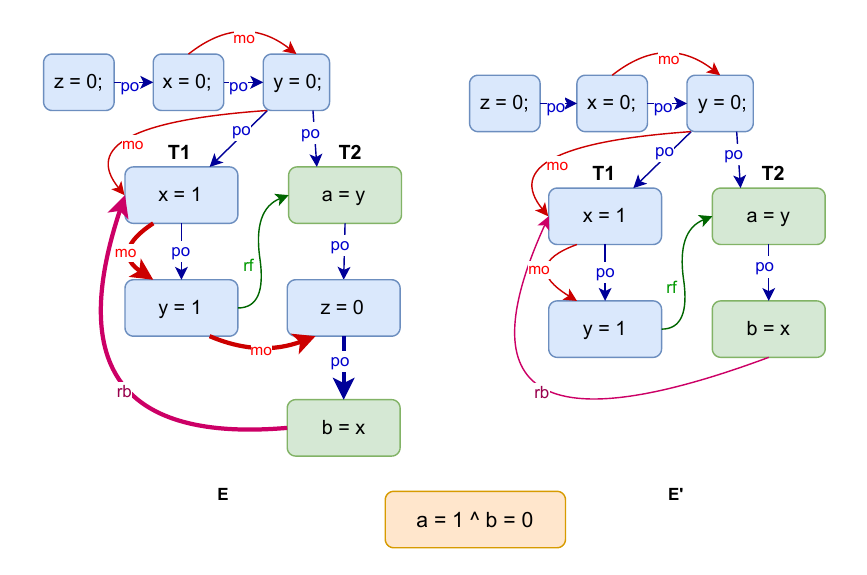}
        \caption{The outcome $a=1 \ \wedge \ b=0$ (yellow box) forbidden under $\textit{SC}_\textit{RR}$ for the original pre-trace due to $E$ but allowed after write-elimination due to $E'$.}
        \label{sc-rr-welim}
    \end{figure*}

    We now have an example where a write-elimination effect is safe under \textit{SC} but unsafe under $\textit{SC}_\textit{RR}$.
    \note{In order to allow such an effect, we may require adding an additional rule to $SC_{RR}$, which will need to prohibit the behavior represented by $E'$.
    However, this additional axiom would be prohibiting the message passing weak outcome $E'$, effectively making reordering of indepdendent reads unsound once again.
    A feasible way to incorporate such effects without adding any axiom is to observe the role of a write in constraints (d) and (e) of $\textit{SC}_\textit{RR}$. 
    We conjecture that it plays an additional role similar to that of constraint (g), that involving $f_{RR}$. 
    While elimination of a write necessarily involves removal of certain $\mo$ relations, those that are to restrict certain outcomes can instead be taken over by an $f_{rr}$ event in place. 
    We keep this as a potential future work.}

    \section{Discussion}
    \label{sec:discussion}
    We have established a representation of program transformations which allow us to decompose them as effects across execution traces.
    This enabled us to reason about safety of transformations w.r.t. memory models in terms of elementary transformation-effects, allowing us to formally specify and reason about compositional property such as Complete, which guarantees safety of program transformations from one model to another.
    We showcased the above formalism's application comparing two models: $\textit{SC}_\textit{RR}$ with respect to \textit{SC}.
    showing $\textit{SC}_\textit{RR}$ Weak, Sound and Complete w.r.t. \textit{SC} for all effects barring write elimination.
     
    \subsection{Usefulness: A New Property For Programming Models}
    
    Def~\ref{def:transf-effect} provides a feasible way to represent transformations.
    Coupled with Theorem~\ref{thm:complete-gen}, we have a process of representing the space of transformation-effects for which their safety can be guaranteed from one model to another.

    \paragraph*{Isolated reasoning of safety for Optimizations}
        One of the most useful property of a memory model is that of data-race-freedom (DRF). 
        It is typically coupled with the fact that under DRF, the model guarantees \textit{SC} reasoning.
        This enables programmers to think about their programs as though it adheres to sequential interleaving, leaving aside much of the complexity of the weak memory model semantics.
        However, DRF-SC may not guarantee a model is Complete w.r.t. \textit{SC}.
        This has implications for compiler optimizers; they cannot solely rely on \textit{SC} reasoning when it comes to designing safe optimizations for the language model.
        Proving models Complete help address this. 
        For instance, if we show the C++ memory model is weaker than that of LLVM, followed by proving C++ Complete w.r.t. LLVM, 
        it guarantees any new optimization designed for LLVM can naturally be applied to programs written in C++.
        This implies LLVM developers need not worry about designing new optimizations that may be rendered unsafe for C++ programs compiled down to its language.
        
    \paragraph*{"DRF-SC" like property for Optimizations}
        Note that DRF is a constraint on the shape of programs, which can be inferred using binary relations in a program's corresponding candidate execution \cite{RaceAdve}. 
        If a model is not Complete w.r.t. another for all transformations, it is possible to identify the set of transformation-effects for which it does.
        Alternatively, this can be translated into a constraint on program shapes, those similar to DRF.
        
        The constraint on transformation-effects we placed (no write-elimination) for proving $\comp{\textit{SC}_{\textit{RR}}}{SC}$ can be equivalently viewed as a constraint on the shape of pre-traces.
        For instance, if the given compiler only performs adjacent write-before-write elimination (eg: $[x=1];[x=2] \mapsto [x=2]$), then for programs that do not have adjacent writes to same location, any optimization done by the compiler under \textit{SC} can also be done under $\textit{SC}_\textit{RR}$.
        As another example, consider Corollary~\ref{cor:sc-rr:rw-reord} as shown below, which states when a transformation-effect involves read-write reordering, when considering $\textit{SC}_\textit{RR}$.
        \begin{restatable}{corollary}{corscrrrwreord}
            \label{cor:sc-rr:rw-reord}
            For transformation-effects quantified by Theorem~\ref{thm:complete-gen} with pair $E \sim E'$ and $B = SC$, $M = SC_{RR}$,
            we have the following for $E$ if no write elimination-effect exists.
            \begin{align*}
                (\rfi;[r_{x}];\po \cup \mo;\rfe;[r_{x}];\po) \ \text{cycle} \implies [r];\po;[w] \in \po^{-}. 
            \end{align*}
            where $r_{x} \in (r(P) \cap r(P'))$.
        \end{restatable}
        Suppose we could not prove $\textit{SC}_\textit{RR}$ Complete w.r.t. \textit{SC} when $E$ had relations as (a) and (c) in Lemma~\ref{lem:sc-rr:comp:general}.
        Then, via Corollary~\ref{cor:sc-rr:rw-reord}, we can conclude that for programs not having read-write syntactic order, we have $\textit{SC}_\textit{RR}$ Complete w.r.t. \textit{SC}.

    \section{Related Work}
    \label{sec:related}
    The earlier implications of memory consistency models on optimizations were observed on those performed by hardware \cite{AdveS}.
    The issue was different, and entailed coming up with a formal specification of hardware memory consistency, much of which were (and still quite a bit) described in informal prose format \cite{OwensS}, \cite{POWERSarkar}.
    An axiomatic formulation, as opposed to operational model greatly helped in clarifying hardware designer intentions \cite{AlglaveJ}.
    The same axiomatic formulation helped in specifying programming language consistency models.

    The implication of language consistency models was only later identified, facing along with informal prose, the non-trivial effects on compilation \cite{PughW}, \cite{BattyS}, \cite{LahavO}, \cite{WattJava}, \cite{GopalECMA}. 
    Earlier criticisms of Java were primarily based on informal/unclear specifications that were hard to understand, naturally leading to most JVMs violating them \cite{PughW2}.
    Optimizations like bytecode reordering, copy propagation, thread-inlining performed on Java Programs violated the specifications \cite{PughW}, \cite{MansonP}.
    A comprehensive study on the impact of optimizations exposed many \textit{SC} optimizations were rendered simply unsafe; examples being redundant read/write eliminations/introductions and reordering independent memory accesses \cite{SevcikJ}. 
    C++ also faced similar issues of informal/unclear specifications, much of which is now properly clarified and formalized \cite{BattyS}.
    However, the detrimental impact on optimizations were being slowly acknowledged, revealing simple optimizations such as register spilling and redundant write elimination unsafe in programs with benign data-races \cite{BoehmC11}.  
    Testing product compilers such as GCC exposed bugs that involved optimizations like redundant eliminations, reordering atomic/non-atomic code fragments and redundant introductions, all of which were unsafe under the C11 memory model.    
    A comprehensive study of the common class of optimizations disallowed in the C11 memory model exposed the consequence of allowing certain causal cycles in conjunction with the restrictions on non-atomic reads on compiler optimizations \cite{VafeiadisV}.
    Validating LLVM optimizations for compiled C11 programs also exposed existing bugs related to shared memory accesses \cite{MetaDataSoham}, some of which are still being discovered \cite{LICM-LLVM}. 
    The JavaScript memory model also was subject to similar issues, with impacts due to unclear specification on hardware mapping and optimizations \cite{WattJava}, \cite{GopalECMA}.
    
    To add, other complications with programming language models involve the infamous out-of-thin-air (OOTA) problem; semantics permitting behaviors which should not be possible by programs \cite{BoehmD}. 
    Recent work concludes that a significant part of this problem can be reduced to tracking semantic-dependencies without disallowing optimizations that remove redundancies \cite{BattyM}.
    Addressing OOTA has resulted in several solutions, yet they do not entirely solve it, permitting behaviors which are still considered OOTA or not successfully able to permit all desired optimizations.
    Java was one of the first languages for which a solution was proposed, involving complex reasoning to justify behaviors involving causal cycles \cite{MansonP}.
    Although this was done catering to optimizations, it still disallows thread-inlining. 
    The \textit{promising semantics} proved to be a good solution; an operational model associating promises to writes which will guarantee to be visible to other threads at a later time in execution \cite{KangPromise}. 
    \note{
    Recent proposals have also seen the use of \emph{event structures}, a representation of programs as a collection of different possible execution paths.
    The Modular Relaxed Dependencies (MRD) model attempts to capture semantic dependencies by calculating a \emph{co-product} from the event structure of the program\cite{PaviottiC}. 
    However, this notion of semantic dependency turns out to be too strict, disallowing optimizations such as thread-inlining, and even certain forms of thread-local redundant read/write eliminations.  
    Another compelling approach is to use an event structure to describe an operational memory model with optimizations included as operational steps\cite{PharabodJ}. 
    These are a small set of per-thread modifications to the event structure that represent thread-local reordering and elimination optimizations.
    These operational steps at the level of event structures coincidentally resemble our notion of (or can be a counterpart to) transformation-effects over pre-traces. 
    However, their formalism is specific to C++11, and it falls short while handling elimination of atomic writes \cite{PichonThesis}.
    While they do have a representation of global optimizations involving value range analysis, other inter-thread analyses that would permit thread-inlining, reordering, elimination and introduction of memory events that would otherwise be unsound is not considered.
    However, we believe their model is appropriate to reason with JIT compilers, something our pre-trace model does not currently handle.
    }

    \note{
    We note that the above solutions using event structure-based models or more generally, multi-execution memory models are to handle the problem of OOTA, as opposed to our goal with pre-traces.
    The pre-trace model is meant to represent a control flow graph of the program and the effects as changes to the graph itself.
    We argue this is closer to a compiler designer's intuition of programs they wish to optimize.  
    Moreover, our model itself is much more conservative w.r.t. the view of program behaviors, permitting more possible outcomes than the program can exhibit (recall example in Fig~\ref{Pr-to-P}, the reads are not associated with any write value yet).
    This conservative view is often desired to design optimizations, and helps in proving properties such as correctness of flow-analyses and portability of optimizations (which in our case, proving Complete between two models).
    While the focus of our work does not inherently deal with OOTA, we strongly believe our approach to decomposing an optimization into observable effects on different traces can help address this issue.
    We \emph{conjecture} the problem of allowing (safe) optimizations involving redundancies can be separated from the OOTA problem, although we leave this for future work.
    Our view is that the idea of event structures flow naturally from trace-level reasoning of programs, whereas the idea of pre-traces and effects flow naturally from reasoning about safety of traditional compiler optimizations via abstract representation of programs.
    It would be unsurprising to us if there exists a convergence point where both these models meet.
    }

    While consistency models such as \textit{SC} do render optimizations unsafe, whether they affect performance of programs is still debatable. 
    A compelling proposal suggests relying just on \textit{SC}, claiming much of the optimizations responsible for performance are already allowed by the model \cite{MarinoSCPresv}. 
    While their proposal arguably simplifies the whole problem, it requires significant changes to hardware, which can enable tracking changes to shared memory state in a concurrent execution.
    Such changes have not been adopted to date by commercial hardware vendors, hinting towards their impracticality.

    Addressing safety of optimizations using execution traces has been done before, but aspects of completeness for which we adopt it has not been addressed. 
    Previous contributions are on optimizations performed on data-race-free programs and those which do not introduce any races in the process \cite{SevcikJ}. 
    Moreover, the space of these optimizations are almost always involving adjacent thread-local events. 
    Our work intends to place no such constraint, and instead guarantee safety of any future optimization from one model to another; be it thread-local or inter-thread global optimizations.  

    Our design of pre-trace matching for identifying correct effects is quite similar to a validator designed to verify C11 optimizations \cite{SohamWeakMO}. 
    In some sense, our work aims generalize this idea of matching in order to create a separation between thread-local (sequential) and concurrent semantics, so that aspects of correctness of transformation can be addressed in a disjoint fashion. 
    An interesting observation from this work is their heuristic approach to loops, claiming utilizing just two iterations suffices in guaranteeing no false positives.
    We intend to investigate this claim in future, albeit from a perspective of identifying the class of memory models for which this claim will hold.  

    Previous work in explaining the interaction between memory consistency and optimizations has been done \cite{LahavV}, showcasing that weak behaviors of \textit{TSO} can be explained via local write-read reordering and read-after-write elimination. 
    It also exposes weak behaviors of models such as release-acquire that cannot be explained via these elementary thread-local transformations nor thread-inlining.
    The results however, are orthogonal to our notion of Complete, and what we refer to as \textit{Optimality}; one must ensure the derived model is not weakened too much that it allows more optimizations than desired.

    \section{Conclusion}
    \label{sec:conclusion}
    In this work, we proposed a representation of a program transformations suitable enough to understand its compositional relation with memory consistency semantics.
    We achieved this by decomposing a transformation into a series of elementary effects on pre-traces, which represent an over-approximation of program behaviors.
    We separated the role of sequential semantics w.r.t. safety of transformations, formalized the desired compositional nature between memory models (\textit{Complete}), followed by showing its practicality by proving $\textit{SC}_\textit{RR}$ Complete w.r.t. \textit{SC}.
    Our contributions expose the need for a new property between memory consistency models, that which may prove useful when designing future optimizations for compilers, thread-local or global.
    Finally, it paves way to a new design methodology for programming language models, one that places emphasis on desired optimizations, with the intention of retaining safety of existing ones. 

    \note{
        We believe our approach is generalizable enough, in that language models can be derived based on a specific set of desired optimizations.
        However, these optimizations must be in the form of syntactic changes that represent a transformation-effect.
        For instance, optimizations like register allocation, algebraic/conditional expression simplifications will be indistinguishable.
        Our approach can be used to allow algorithmic optimizations (eg: Bubble Sort to Merge Sort). 
        But identifying the precise set of effects to enable such optimizations and derive a practical memory model from it may be non-trivial.
        We also assume loops are finite/bounded, i.e. they can simply be unrolled.
        While we can represent loop optimizations (eg: loop-invariant code motion, loop unrolling, etc.) as syntactic effects, asserting their safety w.r.t. unbounded loops will not be possible as of now (Def~\ref{def:transf-safe} would have to assert safety over unbounded set of pre-traces).
    }

    \subsection{Future Work}

    An immediate follow up to our current work would involve addressing elimination effects.
    While we were able to prove $\textit{SC}_\textit{RR}$ Complete barring write eliminations, we suspect the role of writes also play the same role as $f_{rr}$, disallowing any read-read reordering. 
    Thus, a workaround to guarantee write-elimination effects safe from \textit{SC} would be to replace them with $f_{rr}$ instead (any write elimination effect would incur a fence introduction).

    In the future, we want to derive models over \textit{SC} for other forms of elementary reorderings (read-write, write-read, write-write).
    An immediate strong candidate would be write-read reordering, as we can compare it directly to \textit{TSO}, potentially helping us identify the space of transformations for which \textit{TSO} is Complete w.r.t. \textit{SC}.
    We intend to further combine these reordering models, in hopes of identifying a suitable model that allows all forms of reordering, and Complete w.r.t. \textit{SC} for (possibly) all transformation-effects.
    \note{Some of these combinations would closely resemble models like $TSO$, $PSO$ $RMO$ which are known to be used in practice.
    Thus, comparing our derived models with them along with a result of Complete w.r.t. these models will prove beneficial.}
    In order to make these models practical, we would also need to ensure they are weakened \emph{Optimally}, i.e., we do not allow more optimizations than what we desire.
    \note{
        Obtaining a model for the exact desired set may not always be trivial, noting that allowing effects in turn allows multiple desired optimizations. 
        For example, $SC_{RR}$ also permits the optimization $"a=x;c=y;b=x; \ \mapsto \ a=x;c=y;b=a;"$ forwarding the load to $x$ past the read to $y$.
        Forbidding such an optimization would potentially require adding axioms in addition to removing them. 
        To add, proving models Complete may also require manipulation of $\mo$ relations along with $\rf$, especially when inconsistent executions may not have a crucial set (eg: when an execution is inconsistent under $SC$ only due to $\mo;\po$ cycle).
    }  

    Finally, another promising direction would be to address the long-standing the OOTA problem.
    We conjecture that by decomposing optimizations into effects on pre-traces, the aspect of safety can be separated from this problem entirely, allowing compilers to perform optimizations identifying redundancies as desired.

    \appendix

    \newpage

    \section{Auxiliary Observations}
    \label{appendix:A}
    These will be useful for the proofs in Appendix~\ref{appendix:C},~\ref{appendix:D}.

    \note{
    Every memory model we consider guarantees the cycles that violate its consistency rules in an execution are preserved on adding (pending) $\rf$ and $\mo$ relations.
    Formally, 
    \begin{definition}
        \label{def:bin-rel-presv}
        A memory model $M$ is \emph{relation-preserving} if any binary relation $\rel$ of execution $E$ required for the constraints of $M$ is preserved in every $E'$ ($\rel(E) \subseteq \rel(E')$) such that
        \begin{tasks}(3)
            \task $p(E) = p(E')$.
            \task $\rf(E) \subseteq \rf(E')$.
            \task $\mo(E) \subseteq \mo(E')$. 
        \end{tasks}
    \end{definition}
    }

    Crucial sets and piecewise consistency of a model can be used to derive a consistent candidate execution from an inconsistent one.
    \begin{lemma}
        \label{lem:cr-cons}
        For a given a candidate execution $E$ and memory model $M$ such that $\pwc{M}$ and $\neg c_{M}(E)$, if a crucial set $cr$ is non-empty, then there exists a candidate execution $E_{t}$ such that 
        \begin{tasks}(2)
            \task $c_{M}(E_{t})$.
            \task $p(E_{t}) = p(E)$. 
            \task $\po(p(E_{t})) = \po(p(E))$.
            \task $\mo(E_{t}) = \mo(E)$. 
            \task $\rf(E_{t}) \cap \rf(E) = \rf(E) \setminus cr$.  
        \end{tasks} 
    \end{lemma}
    
    \begin{proof}
        Using Def~\ref{def:cr-rf}, we identify crucial set $cr$ and remove the corresponding set of $\rf$ relations to give execution $E'$ such that $c_{M}(E')$ (albeit not well-formed). 
        From $\pwc{M}$, we can use Def~\ref{def:piecewise-cons} to construct $E_{t}$ from $E'$.
    \end{proof}

    The following holds for all transformation-effects not involving any elimination.
    \begin{proposition}
        \label{prop:no-elim-effect}
        For all $P \mapsto_{tr} P'$ such that $st^{-} = \phi$ we have
        \begin{align*}
            \forall E' \in \langle P' \rangle . \exists! E \in \langle P \rangle \ \text{s.t.} \ E \sim E'.
        \end{align*}
    \end{proposition}

    \subsection{Minimal Crucial Sets and Safety}

    \begin{definition}
        \label{def:min-crucial}
        A non-empty crucial set $cr$ for $E$ and memory model $M$ is \textit{minimal} if $\nexists cr' . cr' \subset cr.$  
        Let $Cr(E, M)$ represent the set of minimal crucial sets of candidate execution $E$ for memory model $M$.
    \end{definition}

    We show how minimal crucial sets enable us to assert a transformation-effect without write-eliminations as unsafe.
    \begin{lemma}
        \label{lem:crucial-based-unsafety}
        Consider a memory model $M$ such that $\pwc{M}$.
        Further, consider a pre-trace $P$ and a transformation-effect $P \mapsto_{tr} P'$ that involves no elimination ($st^{-} = \phi$), along with two candidate executions $E \in I\langle P \rangle_{M}, E' \in I\langle P' \rangle_{M}$ such that $E \sim E'$.
        Then, 
        \begin{align*}
            \exists cr' \in Cr(E', M) \wedge Cr(E, M) = \phi \implies \neg \psf{M}{tr}{P}. \\
            \exists cr' \in Cr(E', M), cr \in Cr(E, M) \ \textit{s.t.} \ cr \nsubseteq cr' \implies \neg \psf{M}{tr}{P}. 
        \end{align*}
    \end{lemma}
    
    \begin{proof}
        For both cases, from Lemma~\ref{lem:cr-cons}, we can construct a consistent candidate execution $E'_{t}$ from $E'$ using $cr'$.
        However, from Def~\ref{def:cr-rf}, since $cr$ is non-existent (or minimal), and that $M$ is relation-preserving (Def~\ref{def:bin-rel-presv}), the resultant candidate execution $E_{t} \sim E'_{t}$ from $E$ will still be inconsistent.
        Because $st^{-} = \phi$, from Prop~\ref{prop:no-elim-effect}, no other candidate execution from $\langle P \rangle$ is similar to $E'_{t}$.
        Thus, we have a consistent execution $E'_{t}$ under $M$ in $P'$, for which there does not exist any similar consistent execution in the original pre-trace $P$. 
        From Def~\ref{def:safe-transf}, we can infer $\neg \psf{M}{tr}{P}$.
    \end{proof}
 
    Using Lemma~\ref{lem:crucial-based-unsafety}, we can also infer transformation-effects unsafe for a subset of read-elimination. 
    \begin{corollary}
        \label{cor:read-elim-crucial-unsafety}
        Consider the above setting, but with $st^{-} \cap w(P) = \phi$, with any $cr \in Cr(E, M)$ and $cr_{t} = cr \cap r(P')$ such that $cr_{t} \neq \phi$.
        Then, 
        \begin{align*}
            \exists cr' \in Cr(E', M) \wedge Cr(E, M) = \phi \implies \neg \psf{M}{tr}{P}. \\
            \exists cr' \in Cr(E', M) \ . \ cr_{t} \nsubseteq cr' \implies \neg \psf{M}{tr}{P}. 
        \end{align*}
    \end{corollary}
    
    \begin{proof}
        The first case is same as the first of Lemma~\ref{lem:crucial-based-unsafety}. 
        For the second, we first reassign the $\rf$ for the crucial reads in $cr \setminus cr_{t}$.
        This still keeps both $E$ and $E'$ similar and inconsistent. 
        Without the extra reads, the case now is same as the second case in Lemma~\ref{lem:crucial-based-unsafety}.
    \end{proof}

    \newpage

    \section{Proofs for Sequential Consistency}
    \label{appendix:B}
    \begin{theorem}
        \label{thm:sc-not-redn}
        No consistency rule of \textit{SC} is redundant.
    \end{theorem}

    \begin{proof}

 
        We prove this by contradiction.
        Assume some consistency rule of \textit{SC} is redundant. 
        Then
        \begin{align*}
            \exists a1, a2 \in SC \ \text{s.t.} \ \forall E \ . \ \neg a1(E) \implies \neg a2(E). 
        \end{align*}
        We show that this is not true by construction of a candidate execution that clearly does not satisfy the implication for any given pair of rules.
        We take cases over $a1$, being any of the 5 consistency rules of \textit{SC}, resulting in 5 examples.
        \begin{tasks}(2)
            \task $\mo$ strict total order.
            \task $\mo;\hb$ irreflexive.
            \task $\hb$ irreflexive.
            \task $\rb;\hb$ irreflexive.
            \task $\rb;\mo;\hb$ irreflexive.
        \end{tasks}
        The counter examples are shown in Fig~\ref{sc-cons} as 5 cases.
        For each, notice that exactly one consistency rule of $SC$ is violated, but the others not.
        \begin{figure*}[htbp]
            \centering
            \includegraphics[scale=0.6]{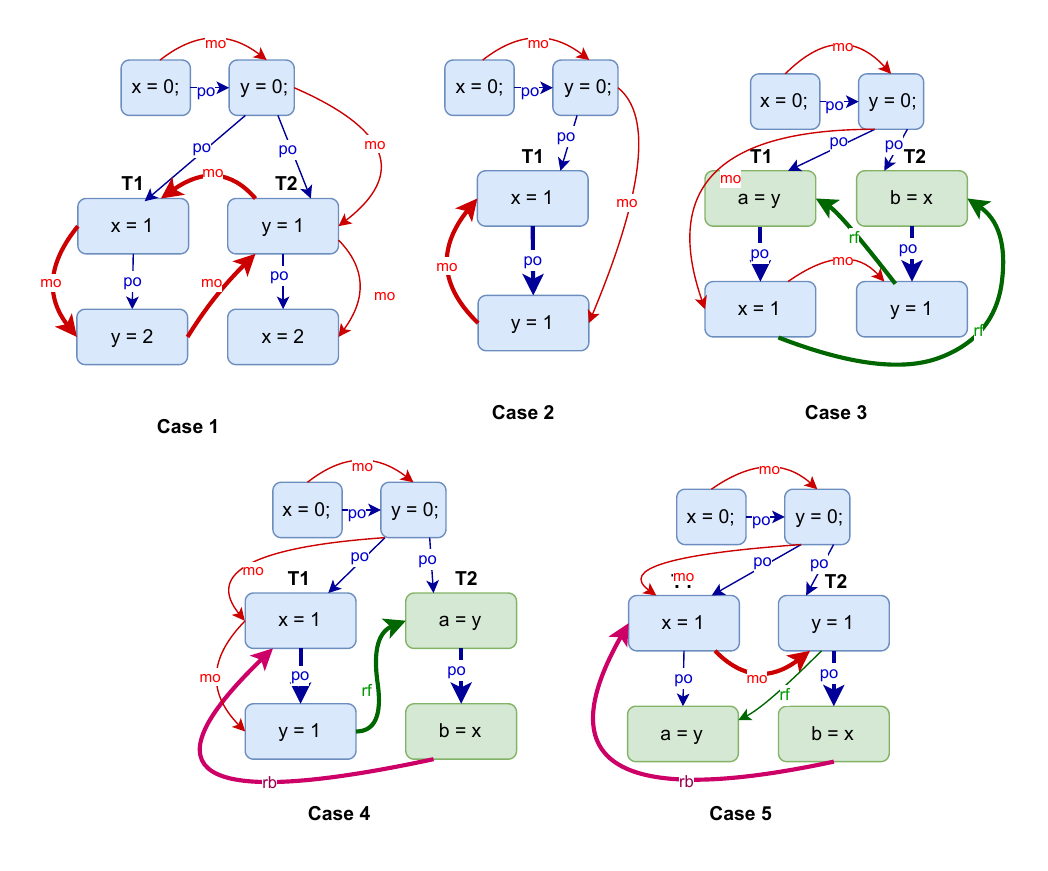}
            \caption{Counter examples for each case of $a1$ being one of the consistency rules.}
            \label{sc-cons}
        \end{figure*}
        
    \end{proof}


    \begin{theorem}
        \label{thm:sc-piecewise-con}
        \textit{SC} is piecewise consistent - $pwc(SC)$.
    \end{theorem}

    \begin{proof}
        

        Consider an execution $E$ such that $c_{SC}(E)$ and that it is not well-formed only due to the missing $\rf$ relation with $r_{x}$.
        The task now is to find a suitable $\rf$ for $r_{x}$ such that none of the consistency rules of $SC$ are violated, giving us our desired candidate execution $E_{t}$.
        We first choose a random write $w_{x}$ and modify our choice based on the cycles formed violating a constraint of $SC$.
        Our modified choice is either an immediate $\mo_{x}$ previous or later write to $w_{x}$.
        We show that this process should terminate, resulting in the desired write. 
        If the first choice violates no rule of $SC$, we are done. 
        If it does, then the following cycles are possible.
        \begin{description}
            \item[Case 1:] $\hb;[w_{x}];\rf;[r_{x}]$ cycle or $\mo;[w_{x}];\rfe;[r_{x}];\hb$ cycle or $\rb;\mo^{?};[w_{x}];\rfe;[r_{x}];\hb$ cycle.
            \begin{itemize}
                \item Choosing an $\mo$ later write $w'_{x}$ ($[w_{x}];\mo;[w'_{x}]$), would always get $[w_{x}];\mo;[w'_{x}];\rf;[r_{x}];\hb$ cycle or $\rb;\mo^{?};[w'_{x}];\rfe;[r_{x}];\hb$ cycle.
                \item Hence, we choose the immediate $\mo$ previous write $w'_{x}$ ($[w'_{x}];\mo;[w_{x}]$) instead.
            \end{itemize}
            \item[Case 2:] $[r_{x}];\rb;[w];\mo^{?};\hb$ cycle 
            \begin{itemize}
                \item Choosing an $\mo$ previous write $w'_{x}$ ($[w'_{x}];\mo;[w_{x}]$), would get the same $[r_{x}];\rb;\mo^{?};\hb$ cycle.
                \item Hence, we choose the immediate $\mo$ later write $w'_{x}$ ($[w_{x}];\mo;[w'_{x}]$) instead.
            \end{itemize}  
        \end{description}
        Since the $\mo_{x}$ has a minimal (initialization) and maximal (candidate execution) element, requiring to choose only $\mo$ previous or only $\mo$ later writes will eventually get us our desired write. 
        However, as we noted above, different cycles may require choosing either one while trying to identify our desired write.
        This can result in never finding a write: choosing an $\mo$ previous one can result in a cycle that requires to choose an $\mo$ later write instead. 
        We show that such a situation would imply $E$ inconsistent under $SC$ without the $\rf$ relation instead, thereby violating our premise.
        There are three cases to consider, pairing the two cycles which would create the scenario mentioned above, while choosing either $w_{x}$ or $w'_{x}$.
        \begin{description}
            \item[Case 1:] $[w];\mo;[w_{x}];\rfe;[r_{x}];\hb$ cycle and $[r_{x}];\rf^{-1};[w'_{x}];\mo_{loc};[w_{x}];\mo^{?};[w'];\hb$ cycle.
            
                This gives us $[w_{x}];\mo^{?};[w'];\hb;\po;[r_{x}];\po;\hb;[w];\mo$ cycle, or $\mo;\hb$ cycle in $E$.
            \item[Case 2:] $[w_{x}];\rf;[r_{x}];\hb$ cycle and $[r_{x}];\rf^{-1};[w'_{x}];\mo_{loc};[w_{x}];\mo^{?};[w'];\hb$ cycle.
            
                This gives us $[w_{x}];\mo^{?};[w'];\hb;\po;[r_{x}];\po;\hb;$ cycle, or $\mo;\hb$ cycle in $E$.
            \item[Case 3:] $\rb;\mo^{?};[w_{x}];\rfe;[r_{x}];\hb$ cycle and $[r_{x}];\rf^{-1};[w'_{x}];\mo_{loc};[w_{x}];\mo^{?};[w'];\hb$ cycle.
            
                This gives us $\rb;\mo^{?};[w_{x}];\mo^{?};[w'];\hb;\po;[r_{x}];\po;\hb$ cycle, or $\rb;\mo^{?};\hb$ cycle in $E$.
        \end{description}
        Thus, we will always be able to find a suitable write $w_{x}$ for $[w_{x}];\rf;[r_{x}]$, resulting in $E_{t}$ being consistent under $SC$.
        Hence, $pwc(SC)$ is true.

        %
    \end{proof}

    \begin{lemma}
        \label{lem:sc-crucial}
        Given an inconsistent execution $E$ under \textit{SC}, a crucial set exists if $\mo$ is a strict total order and does not conflict $\po$.
    \end{lemma}
    \begin{proof}
        Each of the consistency rules of \textit{SC} without the constraints mentioned in the premise has a non-empty $cra$ by Lemma~\ref{lem:sc-cra}.
        $cra(\mo;\hb)$ is non-empty as $\mo$ does not conflict with $\po$.  
        We can thus construct a crucial set $cr$ by taking at least 1 read from each $cra$ set and considering it's $\rf$ in $E$.
    \end{proof}


\newpage

\section{Proofs Regarding \textbf{$\textit{SC}_\textit{RR}$} without Fences and Read-modify-Writes}
    \label{appendix:C}

    \thmrrsound*

    \begin{proof}
        
        Continuing from the proof sketch via contradiction, note that $[r_{x}];\po_{imm};[r_{y}]$ must be a part of the cycle in $E$ violating some rule of $SC_{RR}$.
        Our objective is to show that for each cycle, if $E$ is inconsistent, then so is $E'$ which is comparable to $E$.
        \begin{description}
        \item[Rule (a) violated:] $\mo$ not strict.
        
        The $\po$ between reads does not play a role in $\mo$.
        Thus, $\mo$ remains non-strict even in $E'$, thereby violating Rule (a). 

        \item[Rule (b) violated:] $\hb$ cycle. 
        \begin{flalign*}
            &\hb;[r_{x}];\po_{imm};[r_{y}];\hb \ \text{cycle}. \\ 
            &\hb;[r_{x}];\po_{imm};[r_{y}];\hb. \\ 
            &\hb;\po;[r_{x}];\po_{imm};[r_{y}];\hb \ \vee \ \hb;\rfe;[r_{x}];\po_{imm};[r_{y}];\po. \\
            &\to \hb;\po;[r_{y}];\hb \ \vee \ \hb;\rfe;[r_{x}];\po \ \text{cycle}.  
        \end{flalign*}
        We can see that $\hb$ cycle exists independent of the $\po$ between $r_{x}$ and $r_{y}$.
        Thus, $\hb$ forms a cycle in $E'$, thereby violating Rule (b).
        
        \item[Rule (c) violated:] $\mo;\hb$ cycle.
        \begin{align*}
            &\mo;\hb;[r_{x}];\po_{imm};[r_{y}] \ \text{cycle}. \\ 
            &\mo;\hb;[r_{x}];\po_{imm};[r_{y}];\po. \\ 
            &\mo;\hb^{?};\po;[r_{x}];\po_{imm};[r_{y}];\hb \ \vee \ \mo;\hb^{?};\rfe;[r_{x}];\po_{imm};[r_{y}];\po;\hb^{?}. \\
            &\to \mo;\hb^{?};\po;[r_{y}];\hb \ \vee \ \mo;\hb^{?};\rfe;[r_{x}];\po;\hb^{?} \ \text{cycle}.  
        \end{align*}
        We can see that $\mo;\hb$ cycle exists independent of the $\po$ between $r_{x}$ and $r_{y}$.
        Thus, $\mo;\hb$ forms a cycle in $E'$, thereby violating Rule (c).
    
        \item[Rule (d) or (e) violated:] $\rb;\mo^{?};\hb \setminus a_{rr}$ cycle.
        (We remove $\mo$ for the sake of simplicity.) 
        \begin{align*}
            &\rb;\hb;[r_{x}];\po_{imm};[r_{y}];\hb^{?} \ \text{cycle}. \\ 
            &\rb;\hb;\po;[r_{x}];\po_{imm};[r_{y}];\hb^{?} \ \vee \ \rb;\hb;\rfe;[r_{x}];\po_{imm};[r_{y}];\hb_{?}. \\ 
            &\rb;\hb;[r_{y}];\hb^{?} \ \vee \ \rb;\hb;\rfe;[r_{x}];\po_{imm};[r_{y}];\po;\hb_{?} \ \vee \ a_{rr}. \\ 
            &\to \rb;\hb;[r_{y}];\hb^{?} \ \vee \ \rb;\hb;\rfe;[r_{x}];\po;\hb^{?} \ \vee \ a_{rr} \ \text{cycle}. 
        \end{align*}
        We can see that either $\rb;\hb$ cycle exists independent of the $\po$ between $r_{x}$ and $r_{y}$ or the composition itself becomes part of $a_{rr}$, which is not a constraint of $\textit{SC}_\textit{RR}$.
        Thus, Rule (d) or (e) is violated in $E'$. 
        \end{description}
        
        Hence, via contradiction we have $sound(SC_{RR}, rr)$.
    \end{proof}

    \lemscrrcompgen*

    \begin{proof}
        
        From Theorem~\ref{thm:complete-gen} we know $E \in I\langle P \rangle_{SC_{RR}}$ and $E' \in \llbracket P' \rrbracket_{SC_{RR}}$.
        We also know by $st^- \cap w(P) = \phi$, $\mo$ must be a strict total order for both $E$ and $E'$.
        We go case-wise on the remaining possible rules that could be violated for $E$.
        For each, we show that they reduce to the above claimed compositions which form a cycle.
        Since pre-traces are finite set of events, by monotonicity, this reduction process eventually terminates. 
        To make the reduction process readable, we assume $rr$ to be the composition forming a cycle, thus violating rule $a_{rr}$.
        We also denote the following symbols for the other relations that we need to show forming a cycle in $E$: 
        \begin{tasks}(2)
            \task \textbf{A} - $\rfi;\po$ cycle. 
            \task \textbf{B} - $\mo;\po$ cycle. 
            \task \textbf{C} - $\mo;\rfe;\po$ cycle.
            \task \textbf{D} - $\rb;\mo^{?};\po;[r_{x}]$ cycle.
            \task \textbf{E} - $\rb;\rfe;[r_{x}];\po;[r_{x}]$ cycle. 
        \end{tasks}
    
        \begin{description}
            \item[Rule (b) violated:] $\hb$ cycle.
                \begin{flalign*}
                    &\to \hb \\ 
                    &\to \rf;\po;\hb^{?} \\
                    &\to \rfi;\po \vee \rfe;\po;\hb^{?} \\ 
                    &\to \textbf{A} \vee [w];\rfe;\po;[w'];\rfe;\po;\hb^{?} \\ 
                    &\to \textbf{A} \vee [w'];\mo;[w];\rfe;[r_{x}];\po \vee [w];\mo;[w'];\rfe;\po;\hb^{?} \\ 
                    &\to \textbf{A} \vee \textbf{C} \vee \mo;\hb \ \text{cycle}. 
                \end{flalign*}         
            \item[Rule (c) violated:] $\mo;\hb$ cycle.
                \begin{align*}
                    &\to \mo;\hb \\ 
                    &\to \mo;\hb^{?};\po \\ 
                    &\to \mo;\po \vee \mo;[w];\hb^{?};[w'];\rf;\po \\ 
                    &\to \textbf{B} \vee \mo;[w'];\rf;\po \vee [w];\hb;[w'];\mo \\
                    &\to \textbf{B} \vee \mo;[w'];\rfi;\po \vee \mo;[w'];\rfe;\po \vee \mo;\hb \\ 
                    &\to \textbf{B} \vee \rfi;\po \vee \textbf{C} \vee \mo;\hb \\ 
                    &\to \textbf{B} \vee \textbf{A} \vee \textbf{C} \vee \mo;\hb \ \text{cycle}. 
                \end{align*}
            \item[Rule (d) violated:] $\rb;\hb$ cycle.
                \begin{align*}
                    &\to \rb;\hb \\ 
                    &\to \rb;\hb^{?};\po;[r_{x}] \\ 
                    &\to \rb;\po;[r_{x}] \vee \rb;[w];\hb^{?};[w'];\rf;\po;[r_{x}] \\ 
                    &\to \textbf{D} \vee \rb;[w];\mo;[w'];\rf;\po;[r_{x}] \vee \mo;[w];\hb;[w'] \vee \rb;[w];\rf;\po;[r_{x}] \\ 
                    &\to \textbf{D} \vee \rb;\mo;\rfe;\po;[r_{x}] \vee \rb;\mo;\rfi;\po;[r_{x}] \vee \mo;\hb \vee \rb;[w];\rf;\po;[r_{x}] \\
                    &\to \textbf{D} \vee \textbf{E} \vee rr \vee \rfi;\po \vee \textbf{D} \vee \mo;\hb \vee \rb;\po \vee \rfi;\po \vee \rb;\rfe;\po;[r_{x}] \\ 
                    &\to \textbf{D} \vee \textbf{E} \vee rr \vee \textbf{A} \vee \mo;\hb \vee \textbf{D} \vee \textbf{A} \vee \textbf{E} \\ 
                    &\to \textbf{D} \vee \textbf{E} \vee \textbf{A} \vee rr \vee \mo;\hb \ \text{cycle}.
                \end{align*}
            \item[Rule (e) violated:] $\rb;\mo;\hb$ cycle.  
                \begin{align*}
                    &\to \rb;\mo;\hb \\ 
                    &\to \rb;\mo;\hb^{?};\po;[r_{x}] \\ 
                    &\to \rb;\mo;[w'];\hb^{?};[w];\rf;\po;[r_{x}] \vee \rb;\mo;\po;[r_{x}] \\
                    &\to \rb;\mo;[w'];\mo^{?};[w];\rf;\po;[r_{x}] \vee \mo;[w'];\hb;[w] \vee \textbf{D} \\ 
                    &\to \rb;\mo;[w];\rfe;\po;[r_{x}] \vee \rb;\mo;[w];\rfi;\po;[r_{x}] \vee \mo;\hb \vee \textbf{D} \\ 
                    &\to \rb;\mo;\rfe;\po;[r_{x}] \vee \rb;\mo;\po;[r_{x}] \vee \rfi;\po \vee \mo;\hb \vee \textbf{D} \\ 
                    &\to \rb;\mo;\rfe;[r_{x}];\po;[r_{x}] \vee \rb;\mo;\rfe;[r_{y}];\po;[r_{x}] \vee \textbf{D} \vee \textbf{A} \vee \mo;\hb \vee \textbf{D} \\
                    &\to \rb;\rfe;[r_{x}];\po;[r_{x}] \vee rr \vee \textbf{A} \vee \mo;\hb \vee \textbf{D} \\ 
                    &\to \textbf{E} \vee rr \vee \textbf{A} \vee \mo;\hb \vee \textbf{D} \ \text{cycle}. 
                \end{align*}
        \end{description}

        The above cases show that either $rr$ or one of the relations $\textbf{A, B, C, D, E}$ that form a cycle (by monotonicity) must hold in $E$.
        If only the cycle $rr$ exists in $E$, then by Theorem~\ref{thm:complete-gen}, it is not a valid transformation-effect (as $E$ must be inconsistent under $\textit{SC}_\textit{RR}$). 
        Hence proved.
    \end{proof}

    \thmscrrcompnowelim*

    \begin{proof}
        
        From Lemma~\ref{lem:sc-rr:comp:general}, we know at least one of the following must be true in $E$.
        \begin{tasks}(2)
            \task $\rfi;\po$ cycle. 
            \task $\mo;\po$ cycle.  
            \task $\mo;\rfe;\po$ cycle.  
            \task $\rb;\mo^{?};\po$ cycle. 
            \task $\rb;\rfe;\po$ cycle. 
        \end{tasks}
        
        Consider the above compositions as $rsc$.
        We also know that $E'$ has only $rr$ cycles as opposed to the above ones for $E$.
        From Lemma~\ref{lem:sc-crucial}, we know that a crucial set exists for $E'$.
        From Prop~\ref{prop:sc-rr:cra}, we know $cra(rr) \nsubseteq cra(rsc)$.
        Thus, we can construct a minimal crucial set $cr'$ for $E'$, such that 
        \begin{align*}
            \nexists cr \in Cr(E, SC) \ . \ cr \subseteq cr'.
        \end{align*}

        Since $st^{-} \cap w(P) = \phi$, the $\mo$ cannot change among similar executions.
        Therefore, there can be only more than one $E$ if the eliminated reads are assigned some other $\rf$ relation.
        Depending on the transformation-effect and crucial set $cr$, we have three cases for this 
        \begin{description}
            \item[Case 0:] $\nexists cr \in Cr(E, SC)$
            
                Since $E$ has no crucial set, we trivially have $\neg \psf{SC}{tr}{P}$.

            \item[Case 1:] $\exists cr \ . \ st(r(E')) \cap cr = \phi$.

                This contradicts our premise over $E'$ and $E$ in Theorem~\ref{thm:complete-gen}, as we now can have an execution $E$ consistent under $\textit{SC}_\textit{RR}$ and $E' \sim E$.
                Therefore, this case need not be considered.
            \item[Case 2:] $\exists cr \ . \ st^{-} \cap cr \neq \phi$ 

                We simply remove the reads eliminated from $cr$, giving us $cr_{t}$.
                But note that we would still have $cr_{t} \nsubseteq cr'$ from the above, otherwise it is simply Case 1. 
                From Corollary~\ref{cor:read-elim-crucial-unsafety}, we have $\neg \psf{SC}{tr}{P}$.

            \item[Case 3:] $\exists cr \ . \ st^{-} \cap cr = \phi$

                By Lemma~\ref{lem:crucial-based-unsafety}, we have $\neg \psf{SC}{tr}{P}$.
        \end{description}

        Thus, by contradiction we have $\comp{SC_{RR}}{SC}$.       

    \end{proof}

\newpage

\section{Proofs Regarding $\textit{SC}_\textit{RR}$ on Adding Fences and Read-modify-Writes}

    \label{appendix:D}

    \lemscrrcomprmwfence*

    \begin{proof}
        
        We go case-wise on relations that do bring different possible relations forming a cycle on addition of $rmw$ or $fences$.
        \begin{description}
            \item[Rule (c) violated:] $\mo;\hb$ cycle.
                \begin{align*}
                    &\to \mo;\hb \\ 
                    &\to \mo;\rf \cup \mo;\po;\hb^{?} \\  
                    &\to \mo;\rf \cup \mo;\hb \ \text{cycle}.
                \end{align*}
            \item[Rule (f) violated:] $\rb;\mo$ cycle.
                
                As is. 
            \item[Rule (g) violated:] $\rb;\mo;\rfe;[r_{y}];\po;[f_{rr}];\po;[r_{x \neq y}]$ cycle.            
            
                As is.
        \end{description}
        The above cases show that either $\mo;\hb$ cycle or one of the compositions above (by monotonicity) must hold.
        Thus, at least one of the above mentioned relations may also form a cycle in $E$.

    \end{proof}

    \lemscrrfrr*

    \begin{proof}
        
        From Theorem~\ref{thm:complete-gen} we know $E \in I\langle P \rangle_{SC_{RR}}$ and $E' \in \llbracket P' \rrbracket_{SC_{RR}}$.
        We also know that $\mo$ must be a strict total order for both $E$ and $E'$.
        Consider the relation resulting in $a_{frr}(E) = false$.
        \begin{align*}
            &\rb;\mo^{?};\hb^{?};\rfe;[r_{y}];\po;[f_{rr}];\po;[r_{x \neq y}] \ \text{cycle}.    
        \end{align*}
        We know that the $\rf$ and $\mo$ relations cannot change. 
        Thus, either $\po$ or some $st$ needs to change. 
        We also know that $st^{-}$ cannot involve the reads $r_{x}, r_{y}$ nor any write in $w(P)$.
        Therefore, the only choices remain is to either remove $f_{rr}$ or the two $\po$ relations $[r_{y}];\po;[f_{rr}]$ or $[f_{rr}];\po;[r_{x}]$.
        Hence proved.
    \end{proof}

    \thmscrrcomprmwfrr*

    \begin{proof}
        From Lemma~\ref{lem:sc-rr:comp:rmw-fence}, at least one of the following must be true in $E$.
        \begin{tasks}(2)
            \task $\mo;\rf$ cycle.
            \task $\rb;\mo$ cycle. 
            \task $\rb;\mo;\rfe;\po;[f_{rr}];\po$ cycle.
        \end{tasks}

        First, if the cycle involves the fence $f_{rr}$ like (c), then by Lemma~\ref{lem:sc-rr:frr}, the transformation-effect also involves $f_{rr}$. 
        Since $f_{rr}$ acts as a no-op in \textit{SC}, we do not need to consider such effects for proving Complete.
        We are only concerned with effects on programs with no fences involved.
            
        Now, consider the first two cycles (a), (b) as $usc$.
        We know that $E'$ has only $rr$ cycles as opposed to the above for $E$.
        From Lemma~\ref{lem:sc-crucial}, we know that a crucial set exists for $E'$.
        From Prop~\ref{prop:sc-rr:rmw-frr}, we know $cra(rr) \nsubseteq cra(usc)$.
        Thus, there must be a minimal crucial set $cr'$ such that 
        \begin{align*}
            \nexists cr \in Cr(E, SC) \ . \ cr \subseteq cr'. 
        \end{align*}

        Since $st^{-} \cap w(P) = \phi$, the $\mo$ cannot change among similar executions.
        Therefore, there can be only more than one $E$ if the eliminated reads are assigned some other $\rf$ relation.
        Depending on the transformation-effect and crucial set $cr$, we have three cases for this 
        \begin{description}
            \item[Case 0:] $\nexists cr \in Cr(E, SC)$
            
            Since $E$ has no crucial set, we trivially have $\neg \psf{SC}{tr}{P}$.

            \item[Case 1:] $\exists cr \ . \ st(r(E')) \cap cr = \phi$.
             
                This contradicts our premise over $E'$ and $E$ in Theorem~\ref{thm:complete-gen}, as we know can have an execution $E$ consistent under $\textit{SC}_\textit{RR}$ and $E' \sim E$.
                Therefore, this case need not be considered.
            \item[Case 2:] $\exists cr \ . \ st^{-} \cap cr \neq \phi$. 
             
                We simply remove the reads eliminated from $cr$, giving us $cr_{t}$.
                But note that we would still have $cr_{t} \nsubseteq cr'$ from the above, otherwise it is simply Case 1. 
                From Corollary~\ref{cor:read-elim-crucial-unsafety}, we have $\neg \psf{SC}{tr}{P}$.
            \item[Case 3:] $\exists cr \ . \ st^{-} \cap cr = \phi$.
            
                By Lemma~\ref{lem:crucial-based-unsafety}, we have $\neg \psf{SC}{tr}{P}$.
        \end{description}

        Thus, by contradiction, for the models \textit{SC} and $\textit{SC}_\textit{RR}$ with the inclusion of $f_{rr}$ and $rmw$ events, we have $\comp{SC_{RR}}{SC}$.      
        
    \end{proof}

    \corscrrrwreord*

    \begin{proof}
        By Lemma~\ref{lem:sc-rr:comp:general} we know at least one of the mentioned relations form a cycle.
        We also know that none of these cycles compositions exist in $E'$.
        Therefore, we only need to check what we can remove in order to eliminate the cycle. 
        Let us go case-wise: 
        \begin{description}
            \item[Case 1:] $\rfi;[r_{x}];\po$ cycle.

                Since $\rfi;[r_{x}]$ cannot change, only the $\po$ relation can change to eliminate the cycle.  
                Since we have $[r_{x}];\po;[w_{x}]$, we need $[r_{x}];\po;[w_{x}] \in \po^{-}$
                
            \item[Case 2:] $[w_{y}];\mo;\rfe;[r_{x}];\po$ cycle.

                Since both $\mo$ and $\rfe$ cannot change from the premise, we have to change $\po$ to eliminate the cycle.
                Since we have $[r_{x}];\po;[w_{y}]$, we need $[r_{x}];\po;[w_{x}] \in \po^{-}$.
        \end{description}
    
        Thus, read-write de-ordering is part of the transformation-effect.
    \end{proof}

    \newpage

    \bibliography{ref.bib}

  \end{document}